\documentclass[letterpaper, 10 pt, journal]{ieeetran}    

\usepackage{graphics} 
\usepackage{epsfig} 
\usepackage{amsmath} 
\usepackage{amssymb}  
\usepackage{amsthm}
\usepackage{upgreek}
\usepackage{dsfont}
\usepackage{subfigure}
\usepackage{cancel}
\usepackage{mathtools}
\usepackage{multirow}  
\usepackage{enumitem}
\usepackage{color}
\usepackage{ragged2e,array,booktabs}
\usepackage{lipsum} 
\usepackage{cuted}
\usepackage{array}
\usepackage[numbers,sort&compress]{natbib}
\usepackage{color}
\usepackage{algorithm}
\usepackage{algpseudocode}

\newtheorem{thm}{Theorem}
\newtheorem{proposition}{Proposition}
\newtheorem{corollary}{Corollary}
\newtheorem{lemma}{Lemma}

\newtheorem{definition}{Definition}
\newtheorem{remark}{Remark}
\newtheorem{example}{Example}

\newtheorem{note}{Note}

\newcommand\norm[1]{\left\lVert#1\right\rVert}
\newcommand\E{\mathcal{E}}
\newcommand\EE{\mathbb{E}}
\newcommand\V{\mathbb{V}}
\newcommand\R{\mathbb{R}}

\definecolor{LG}{rgb}{0.13, 0.55, 0.13}

\renewcommand{\baselinestretch}{0.98}

\usepackage{todonotes}
\usepackage{stmaryrd}

\title{A Geometric Method for Passivation and Cooperative Control of Equilibrium-Independent Passive-Short Systems}

\author{Miel Sharf,~\IEEEmembership{Graduate~Student~Member,~IEEE}, Anoop Jain,~\IEEEmembership{Member,~IEEE}, and Daniel Zelazo,~\IEEEmembership{Senior~Member,~IEEE}
\thanks{M. Sharf is with the Division of Decision and Control Systems, KTH Royal Institute of Technology, Stockholm, Sweden. {\tt\small sharf@kth.se}.   A. Jain is with the Department of Electrical Engineering, Indian Institute of 
    Technology, Jodhpur, India. {\tt\small anoopj@iitj.ac.in}. D. Zelazo is with the Faculty of Aerospace Engineering, Israel Institute of Technology, Haifa, Israel {\tt\small dzelazo@technion.ac.il}. This work 
    was supported in part at the Technion by Lady Davis Fellowship, and the 
    German-Israeli Foundation for Scientific Research and Development.}}

\begin{document}
\maketitle
		
		


\begin{abstract}
Equilibrium-independent passive-short (EIPS) systems are a class of systems that satisfy a passivity-like dissipation inequality with respect to any forced equilibria with non-positive passivity indices. This paper presents a geometric approach for finding a passivizing transformation for such systems, relying on their steady-state input-output relation and the notion of projective quadratic inequalities (PQIs). We show that PQIs arise naturally from passivity-shortage characteristics of an EIPS system, and the set of their solutions can be explicitly expressed. We leverage this connection to build an input-output mapping that transforms the steady-state input-output relation to a monotone relation, and show that the same mapping passivizes the EIPS system. We show that the proposed transformation can be implemented through a combination of feedback, feed-through, post- and pre-multiplication gains. 
Furthermore, we consider an application of the presented passivation scheme for the analysis of networks comprised of EIPS systems. 
Numerous examples are provided to illustrate the theoretical findings.
\end{abstract}




\vspace{-5pt}
\section{Introduction}
Cooperative control has been extensively studied in the last few years, as it displays both interesting theoretical questions, as well as a wide range of engineering applications \cite{OlfatiSaber2007,Oh2015,Hatanaka2015}. One widespread tool in cooperative control is the notion of passivity \cite{DePersis2018,Antsaklis2013,Hatanaka2015}. Passivity theory was first applied to multi-agent systems in \cite{Arcak2007}, where it was used to solve group coordination problems. Since then, different variants of passivity were used for solving various problems in robotics \cite{Chopra2006}, synchronization 
\cite{Stan2007}, and distributed optimization \cite{Tang2016}.

The classical notion of passivity, as appears in \cite{Khalil2002}, is defined with respect to equilibrium at the origin. Some authors also define shifted passivity, which is defined with respect to an input-output (I/O) pair of the system, to apply passivity-based methods to systems having forced equilibria \cite{Arcak2007,Monzshizadeh2019,SimpsonPorco2019}. For brevity, we shall not differentiate the two concepts, and refer to both as passivity. The notion of passivity with respect to a single input-output pair may not be sufficient for stability analysis of multi-agent systems, as the interconnection of (shifted)-passive systems is stable only if the closed-loop network has an equilibrium, which can be hard to verify for networks comprised of multiple nonlinear agents having different dynamics.  

To remedy this issue, several variants of passivity were developed, demanding systems to be passive with respect to any equilibrium input-output pairs or trajectories. Incremental passivity \cite{Pavlov2008} demands that a passivation inequality is held with respect to pairs of trajectories, but is often too restrictive. Another variant, equilibrium-independent passivity (EIP), demands that the system is passive with respect to any equilibrium it has, and models the steady-state output as a continuous (monotone) function of the steady-state input \cite{Hines2011,SimpsonPorco2019}. This variant has many applications, e.g. \cite{Meissen2015,SimpsonPorco2016}, but does not include some fundamental systems such as the single integrator, characterized by having multiple steady-state outputs for the steady-state input $\mathrm u=0$ (due to different initial conditions). Another variant of passivity is maximal equilibrium-independent passivity (MEIP), introduced in \cite{Burger2014}. Here, passivity is assumed with respect to all equilibria, and the steady-state output is modeled as a maximally monotone relation of the steady-state input, generalizing EIP. In \cite{Burger2014}, it was shown that a diffusively-coupled network of SISO output-strictly MEIP agents and SISO MEIP controllers converges, and its limit can be found as the minimizers of two dual convex network optimization problems associated with the network, usually referred to as the optimal flow problem and optimal potential problem \cite{Rockafellar1998}. In this way, the convex network optimization problems give a computationally viable way of computing the limit of the diffusively-coupled network. This connection was used in \cite{Sharf2017,Sharf2018a, Sharf2019e} to solve various synthesis problems, and in \cite{Sharf2019f} for fault detection and isolation problems. 

In practice, however, many systems are not passive \cite{Qu2014,Harvey2016,Trip2018,Xia2014}. Their lack of passivity is often quantified using the input-passivity index and the output-passivity index \cite{Zhu2014}, and is often compensated using passivation methods (also known as passification methods \cite{Fradkov2003}). The goal of this paper is to present a novel passivation method for systems which are not passive, but have a shortage of passivity, characterized by a weaker dissipation inequality.

\subsection{Literature Review}
The most common methods to passivize a system rely on feedback. A well-known approach is output-feedback using a fixed gain \cite{Khalil2002}. This approach passivizes systems with a negative output-passivity index \cite{Zhu2014}, otherwise known as output passive-short systems.  Another method considers output-feedback using a controller with prescribed passivity indices \cite{Zhu2014}, but passivation is again achieved only for passive-short systems \cite[Theorem 7]{Zhu2014}. One can similarly consider input-feedthrough, passivizing systems with a negative input-passivity index \cite{Zhu2014}, known as input-passive-short systems.

Other prominent feedback-based methods used for passivation include state-feedback and output-feedback by general static nonlinearities, see \cite{Byrnes1989,Byrnes1991,Fradkov1995,Jiang1996,Fradkov1998,Fradkov2003} and references therein. 
These approaches were proven to work for weakly minimum phase systems with relative degree at most $1$, but can have several problems. First, like Lyapunov theory, these methods are often non-constructive, and heavily rely on structural properties of the system at hand \cite[Chapter 1]{Sepulchre2012}. Second, the construction of the feedback law requires an exact model of the system, or at least an approximate one. This can be a problem in cases where the model of the system changes,  due to faults, wear-and-tear, unforeseen working conditions, etc. As passivity indices can be estimated using in-run data \cite{Romer2019,Montenbruck2016,Romer2017a}, passivation methods relying on passivity indices can mitigate this effect by adapting the assumed passivity indices. We also mention other methods building on state-feedback, such as backstepping and forwarding \cite[Chapter 6]{Sepulchre2012}, which remove either the minimum-phase or the relative-degree requirement, but replace it with a structural assumption on the model of the system, i.e., the system must be in a triangular form.

A novel method for mitigating the problems of feedback-based methods was presented in \cite{Xia2018}. The method considers a  general I/O transformation, which defines a new input and a new output for the system as a linear combination of its original input and output. This method generalizes output-feedback and input-feedthrough with constant gains. In \cite{Xia2018}, this I/O transformation was used to passivize systems with a finite $\mathcal{L}_2$-gain. Namely, the entries of the matrix defining the I/O transformation were chosen according to the $\mathcal{L}_2$-gain of the system at hand by solving a collection of equations and inequalities. In particular, the method is constructive and can successfully cope with a change in the dynamics by measuring the $\mathcal{L}_2$-gain of the new system and updating the entries of the matrix accordingly. However, the applicability of this method is limited to systems with a finite $\mathcal{L}_2$-gain, which excludes all unstable systems, input- or output-passive short systems, as well as some marginally stable systems such as the single integrator. Thus there is a need for a more sophisticated passivization approach to deal with a wider class of systems. This motivates the goals of this paper.

\subsection{Contributions}

In this paper, we build on \cite{Xia2018} and propose a novel method for constructing passivizing I/O transformations. Our approach is based on analytic geometry, which is applicable to a wider class of systems characterized by a passivity-like dissipation inequality with arbitrary passivity indices. Unlike in \cite{Xia2018}, these systems need not have a finite $\mathcal{L}_2$-gain. We define these systems as input-output $(\rho,\nu)$-passive systems, including, but not restricted to, output passive-short system, input passive-short systems and finite $\mathcal{L}_2$-gain systems. We show how to use the passivity indices of such systems to build a passivizing I/O transformation that can be realized using an amalgamation of easily implementable components such as input-feedthrough, output-feedback, and gains. We consider systems that are input-output $(\rho,\nu)$-passive with respect to all forced equilibria. The collection of all these steady-state input-output pairs is known as the steady-state I/O relation of the system. The steady-state I/O relation for passive systems is known to be monotone \cite{Hines2011, Burger2014}, and we show that this relation is non-monotone for passive-short systems. To tackle such systems, we introduce the notion of projective quadratic inequalities (PQIs), that are inequalities in two scalar variables, as well as methods inspired from analytic geometry to find a linear transformation monotonizing\footnote{We introduce this word and it has the meaning of ``to make monotone." In simple words, monotonizing means converting any (non-monotone) relation to a monotone relation.} the steady-state relation of the system. We then show that the linear transformation gives rise to an I/O transformation, which is shown to passivize the system with respect to all forced equilibria. We further discuss an application of this passivation scheme for multi-agent systems, in which, the notion of MEIP leads to a network optimization framework for analysis. As we already know that the passivized systems have monotone steady-state relations, the missing key notion for assuring MEIP is maximality. In this direction, we introduce the notion of cursive relations to assert maximality of the monotonized relations, proving the agents are MEIP, and allowing us to derive a transformed network optimization framework in the spirit of \cite{Burger2014}. We also reproduce the results of \cite{Jain2018} as a special case, which proves a network optimization framework assuming the agents only have an output-shortage of passivity. We exemplify our results by characterizing a class of linear and time-invariant systems as EIPS systems, and give two case studies by comparing our results with the existing literature. 
We emphasize that our results are also valid for classical passivity, as PQIs abstract all notions of classical passivity discussed in the introduction.

The rest of the paper is organized as follows. Section~\ref{sec_Background and Problem Formulation} presents some background and provides a few definitions. Section~\ref{subsec.ProbForm} motivates and formulates the problem. Section~\ref{Monotonization} discusses the steady-state I/O relation of passive-short systems, and suggests a geometric method of finding a monotonizing transformation. Section~\ref{Passivation} shows that the monotonizing transformation passivizes the system, and shows how to implement the said transformation using basic control elements, such as feedback, feed-through, and gains. Section~\ref{Finite L2 Gain} discusses the notion of input-output $(\rho,\nu)$-passivity and its generality. Section~\ref{Cursive} studies the last obstacle needed for MEIP, namely \emph{maximal} monotonicity, and formulates the network optimization framework. Section \ref{An Example} presents two examples of applying our methods, before we conclude the paper in Section \ref{conclusions}.

\paragraph*{Preliminaries}
We use notions from graph theory \cite{Godsil2001}. A graph is a pair $\mathcal{G} = (\mathbb{V}, \mathbb{E})$, consisting of a finite set of vertices $\mathbb{V}$, and a finite set of edges, $\mathbb{E} \subset \mathbb{V} \times \mathbb{V}$. Each edge $e \in \mathbb{E}$ consists of two vertices $i, j \in \mathbb{V}$, and the notation $e = (i, j)$ indicates that $i$ is the \emph{head} of edge $e$ and $j$ is its \emph{tail}. The incidence matrix $\E \in  \mathbb{R}^{|\V|\times|\EE|}$ of $\mathcal{G}$ is defined such that for any edge $e = (i, j)$, $[\E]_{ie} = +1, [\E]_{je} = -1$, and $[\E]_{\ell e} = 0$ for $\ell \neq i, j$. The $n \times n$ identity matrix is denoted by $\mathrm {Id}_n$, and $\pmb{0}_n$ is the all-zero vector. The Legendre transform of a convex function $\varPhi:\R^d\to\R$ is a function $\varPhi^\star:\R^d\to \R$ defined by $\varPhi^\star(y) = \sup_{u\in \R^d}\{u^\top y - \varPhi(u)\}$ \cite{Rockafellar1997}. Moreover, the subdifferential of a convex function $\varPhi$ is denoted as $\partial \varPhi$. A relation, i.e., a subset $\Omega \subseteq \mathcal{A} \times \mathcal{B}$ of a product set, is identified with the set-valued map sending  $a \in \mathcal{A}$ to $\{b\in\mathcal{B}:\ (a,b)\in \Omega\}$. Given a relation ${\Omega} \subseteq \mathcal{A} \times {\mathcal{B}}$, ${\Omega}^{-1}$ denotes the inverse relation of ${\Omega}$, i.e., ${\Omega}^{-1} := \{(b, a) \in \mathcal{B} \times \mathcal{A} : (a, b) \in {\Omega}\}$. We follow the convention that italic letters denote dynamic variables and letters in normal font denote constant signals. 

\section{Background}\label{sec_Background and Problem Formulation}
This section reviews the concept of MEIP, introduces systems with finite 
equilibrium-independent passivity indices, and describes the network model for diffusively coupled systems. 
\subsection{Maximal Equilibrium-Independent Passivity}
Consider the following SISO dynamical system,
\begin{equation}\label{system_model_single}
\Upsilon:~ \dot{x} = f(x, u);~~~{y} = h(x, u),
\end{equation}
with state $x \in \mathbb{R}^{n}$, control input $u \in \mathbb{R}$ and output $y \in \mathbb{R}$. The functions $f$ and $h$ are assumed to be sufficiently smooth. We assume the systems in the form \eqref{system_model_single} admit forced steady-state input-output equilibrium pairs.  This leads to the following definition, used extensively in the literature \cite{Burger2014,Sharf2018a,Hines2011,SimpsonPorco2019}.

\begin{definition}
The \emph{steady-state input-output relation} of the system \eqref{system_model_single} is the collection of all steady-state input-output pairs $(\mathrm{u,y})$. That is, it is equal to the set $k = \{(\mathrm{u,y}):\ \exists \;\mathrm x, \,\pmb{0}_n = f(\mathrm{x,u}),\ \mathrm y = h(\mathrm{x,u})\}$. The corresponding inverse relation is given by $k^{-1} = 
\{(\mathrm{y,u}): (\mathrm{u,y}) \in k\}$.
\end{definition}

Note that any steady-state relation can be thought of as a set-valued map. Namely, for any constant input $\mathrm u$, we can define $k(\mathrm u)$ as the set $k(\mathrm u) = \{\mathrm y:\ (\mathrm{u,y}) \in k\}$. Note that $k(\mathrm u) = \emptyset$ if no steady-state output corresponding to the input $\mathrm u$ exists. Similarly, for a steady-state output $\mathrm y$, we define $k^{-1}(\mathrm y)$ as $k^{-1}(\mathrm y) = \{\mathrm u:\ (\mathrm{u,y}) \in k\}$, the set of all constant inputs $\mathrm u$ that can generate $\mathrm y$. In this sense, the inverse relation can always be defined, as we do not assume $k$ to be a function.

For EIP systems, it is shown in \cite{Hines2011} that the steady-state I/O relation $k$ is a continuous and monotonically increasing function. In particular, for any steady-state input ${\rm u}$ there is exactly one steady-state output ${\rm y}$.  However, EIP excludes some important system classes, e.g. the single integrator \cite{Burger2014}. To capture the behavior of systems where the steady-state I/O relations are not necessarily a function, but rather a \emph{relation}, the notion of MEIP was suggested relying on \emph{maximal monotonicity} of the steady-state I/O relation \cite{Burger2014}. 

\begin{definition}
A relation $k$ is said to be \emph{maximal monotone} if
\begin{enumerate}
\item[i)] it is monotone, i.e., for any $({\rm u}_1, {\rm y}_1), ({\rm u}_2, {\rm y}_2) \in k$, we have that $({\rm u}_2 - {\rm u}_1) ({\rm y}_2 - {\rm y}_1) \geq 0$, and
\item[ii)] it is not contained in a larger monotone relation.
\end{enumerate}
\end{definition}
The notion of maximal monotonicity is closely related to convex functions as described in the following theorem.

\begin{thm}[\hspace{-.1pt}\cite{Rockafellar1997}]\label{thm_monotone_relation_convex_function}
A relation $k$ is maximally monotone if and only if there exists a convex function $\Phi$ such that the subgradient $\partial \Phi$ is equal to $k$. Moreover, $\Phi$ is unique up to an additive constant. The function $\Phi$ is called the \emph{integral function} of $k$.
\end{thm}

Maximal monotonicity induces the following system-theoretic property:
\begin{definition}[\hspace{-.1pt}\cite{Burger2014}] \label{def_MEIP}
A dynamical SISO system ${\Sigma}: u \mapsto y$ is (output-strictly) \emph{maximal equilibrium independent passive (MEIP)} if
\begin{enumerate}
\item[i)] The system $\Sigma$ is (output-strictly) passive with respect to any steady-state I/O pair $(\mathrm{u,y})$ it possesses.
\item[ii)] The associated steady-state I/O relation is maximally monotone.
\end{enumerate}
\end{definition}

Examples of MEIP systems include single integrators, port-Hamiltonian systems, gradient systems, and others; see \cite{Burger2014} for further discussion. One important aspect of MEIP systems is their integral functions, as mentioned in Theorem~\ref{thm_monotone_relation_convex_function} above. 
Since the steady-state I/O relation $k$ is maximally monotone for an MEIP 
system, there exists a convex function $K$ such that $\partial K = k$. 
Moreover, the Legendre transform of $K$, denoted as $K^\star$, is also a convex 
function, and satisfies $\partial K^\star = k^{-1}$. Thus both $k,k^{-1}$ 
have integral functions that are necessarily convex. However, this is not true 
for passive-short systems, as will be shown in Section \ref{subsec.ProbForm}.

\subsection{Equilibrium-Independent Shortage of Passivity}
The main advantage of applying an equilibrium-independent notion of passivity 
for multi-agent systems is that it allows to prove convergence without 
specifying the steady-state limit (see 
\cite{Burger2014,Hines2011,SimpsonPorco2019} and 
Subsection \ref{subsec.DiffCoup}). However, many systems in practice are not 
passive \cite{Qu2014,Harvey2016,Trip2018,Xia2014}, and even fewer are passive 
with respect to all equilibria. The level of passivity, or shortage thereof, is 
usually measured using passivity indices. We first define the notion of shortage of passivity that we consider, and later adjust it to fit into the equilibrium-independent framework.

\begin{definition} \label{defn_inpOutPassive}
Let $\Sigma$ be a SISO system with a constant input-output steady-state pair $(\mathrm{u,y})$. The system $\Sigma$ is said to be:
\begin{enumerate}
\item[i)] \emph{output $\rho$-passive} with respect to $(\mathrm{u,y})$ if there exist a storage function $S(x)$, and a number $\rho \in \R$, such that the following inequality holds for any trajectory:
\begin{equation}\label{OSMEIP}
\dot{S} \leq -\rho (y - {\rm y})^2 + (y - {\rm y})(u - {\rm u});
\end{equation}
\item[ii)] \emph{input $\nu$-passive} with respect to $(\mathrm{u,y})$ if there exist a storage function $S(x)$, and a number $\nu \in \R$, such that the following inequality holds for any trajectory:
\begin{equation}\label{ISMEIP}
\dot{S} \leq -\nu (u - {\rm u})^2 + (y - {\rm y})(u - {\rm u});
\end{equation}
\item[iii)] \emph{input-output ($\rho,\nu$)-passive} with respect to $(\mathrm{u,y})$ if there exist a storage function $S(x)$, and numbers $\rho,\nu \in \R$, such that $\rho\nu < \frac{1}{4}$ and that the following inequality holds for any trajectory:
\begin{equation}\label{IOSMEIP}
\dot{S} \leq -\rho (y - {\rm y})^2 -\nu (u - {\rm u})^2 + (y - {\rm y})(u - {\rm u}).
\end{equation}
\end{enumerate}
\end{definition}

\begin{remark}
Output $\rho$-passive systems with $\rho<0$ are known in the literature both as output-passive short or output passivity-short systems \cite{Qu2014,Joo2016,Harvey2016,Atman2018,Jain2018,Sharf2019c} or as output-passifiable systems \cite{Bondarko2003,Selivanov2016}. Similarly, input $\nu$-passive systems with $\nu<0$ are usually called input-passive short systems or as input-passifiable systems.
\end{remark}

\begin{definition} \label{defn_passive_short}
A SISO system $\Sigma : u\mapsto y$ is said to be:
\begin{enumerate}
\item[i)] \emph{Equilibrium-Independent Output $\rho$-Passive} (EI-OP($\rho$)) if it is output $\rho$-passive with respect to any equilibrium.
\item[ii)] \emph{Equilibrium-Independent Input $\nu$-Passive} (EI-IP($\nu$)) if it is input $\nu$-passive with respect to any  equilibrium.
\item[iii)] \emph{Equilibrium-Independent Input-Output $(\rho,\nu)$-Passive} (EI-IOP($\rho,\nu$)) if it is input-output ($\rho,\nu$)-passive with respect to any equilibrium.
\end{enumerate}
Moreover, for EI-OP($\cdot$) and EI-IP($\cdot$), the largest numbers $\rho,\nu$ for which the inequalities \eqref{OSMEIP} and \eqref{ISMEIP} hold are called the \emph{equilibrium-independent output-passivity index} and \emph{equilibrium-independent input-passivity index} of the system, respectively. Furthermore, $\Sigma$ is said to be \emph{equilibrium-independent passive short} (EIPS) if there exist $\rho,\nu$ with $\rho\nu < \frac{1}{4}$ such that $\Sigma$ is EI-IOP($\rho,\nu$).
\end{definition}

\begin{remark}
The numbers $\rho,\nu$ in Definition \ref{defn_passive_short} are not unique, as decreasing them makes the inequality easier to satisfy. We thus define the equilibrium-independent passivity indices analogously to the output-feedback passivity index (OFP) and the input-feedthrough passivity index (IFP) in \cite{Xia2014}. Moreover, the definition above unites strictly-passive, passive, and passive-short systems. The case $\rho,\nu>0$ corresponds to strict passivity, $\rho,\nu = 0$ corresponds to passivity, and $\rho,\nu < 0$ corresponds to shortage of passivity. Thus, it will allow us to consider networks of systems where some are passive and some are passive-short, without needing to specify the exact passivity assumption. It also allows us to consider EI-IOP($\rho,\nu$) systems for $\rho > 0$ and $\nu < 0$ (or vice versa) with no additional effort needed.
\end{remark}
\begin{remark}
The demand that $\rho \nu < \frac{1}{4}$ for defining EI-IOP($\rho,\nu$) might seem unnatural. The reason we add it is that otherwise, the right-hand side of \eqref{IOSMEIP} will either be always positive or always negative. The first case implies all static nonlinearities are EI-IOP($\rho,\nu$), and the second case implies that no system can be EI-IOP($\rho,\nu$), both rendering the definition useless.
\end{remark}

\begin{remark}
EI-IOP($\rho,\nu$) systems capture both EI-OP($\rho$) and EI-IP($\nu$) systems by setting either $\rho =0$ or $\nu=0$.
\end{remark}
We now give an example of a class of EI-OP($\rho$) systems:
\begin{proposition}\label{prop_gradient_system}
Consider the SISO gradient system $\dot{x} = -\nabla U (x) + u; y=x$, where the 
Hessian of the potential $U$ satisfies $\mathrm{Hess}(U) \geq 
\rho{\rm Id}$ for some $\rho \in \mathbb{R}$. Then $\Sigma$ is 
EI-OP($\rho$).    
\end{proposition} 

\begin{proof}
Take a steady-state I/O pair $(\mathrm{u,y})$ and note $\mathrm x=\mathrm y$ is the corresponding state at equilibrium.
Consider the storage function $S(x) = \frac{1}{2}\norm{x - {\rm x}}^2$. The derivative of $S$ along the system trajectories is $\dot{S} = (x - {\rm x})^\top(-\nabla U(x) + u)$. Defining $\varphi(x) := \nabla U(x) - \rho x$, we write $\dot{S} = (x - {\rm x})^\top(-\varphi(x) - \rho x + u)$. Adding and 
subtracting $\varphi({\rm x})$ and $\rho {\rm x}$ and using the fact that ${\rm 
u} = \nabla U({\rm x}), {\rm y} = {\rm x}$ and $\varphi({\rm x}) = \nabla 
U({\rm x}) - \rho {\rm x}$ at equilibrium, we obtain $\dot{S} = -(x - {\rm 
x})^\top((\varphi(x) - \varphi({\rm x}))  - \rho (y - {\rm y})^\top(y - {\rm y}) + (y 
- {\rm y})(u - {\rm u}))$. It is straightforward to verify that $\mathrm{Hess}(U) \geq \rho {\rm Id}$ implies that $\nabla \varphi(x) \geq 0$, so $\varphi(\cdot)$ is a monotone operator, that is, $-(x - {\rm 
x})^\top((\varphi(x) - \varphi({\rm x})) \leq 0$. We thus conclude that 
$\dot{S} \leq - \rho (y - {\rm y})^\top(y - {\rm y}) + (y - {\rm y})^\top(u - {\rm 
u}))$, and hence the system is EI-OP($\rho$).       
\end{proof}

\subsection{Diffusively-Coupled Network Model} \label{subsec.DiffCoup}
We consider a collection of SISO agents interacting over a network 
$\mathcal{G}=(\V,\EE)$, in which the agents reside at the nodes $\V$, and the 
edges regulate the relative output between the associated nodes. Namely, the 
agents $\{\Sigma_i\}_{i\in\V}$ and the controllers $\{\Pi_e\}_{e\in\EE}$ have 
the following models:
\begin{align}
\Sigma_i:\ \begin{cases} \dot{x}_i = f_i(x_i,u_i) \\ y_i = h_i(x_i,u_i)\end{cases}, \Pi_e:\ \begin{cases} \dot{\eta}_e = \phi_e(\eta_e,\zeta_e) \\ \mu_e = \psi_e(\eta_e,\zeta_e)\end{cases},
\end{align}
where $x_i\in \R^{\ell_i},\ \eta_e \in \R^{\ell_e}$ are the states, $u_i, \zeta_e \in \R$ are the inputs and $y_i,\mu_e$ are the outputs. We define the stacked vectors $\pmb{u} = [u_1,\cdots,u_{|\V|}]^\top$, and similarly for $\pmb{x},\pmb{y},\pmb{\zeta},\pmb{\eta}$ and $\pmb{\mu}$. The agents and controllers are coupled by $\pmb{\zeta} = \E^\top \pmb{y}$ and $\pmb{u} = -\E\pmb{\mu}$, where $\E$ is the incidence matrix of $\mathcal{G}$. The closed-loop system is called the \emph{diffusively-coupled system} $(\pmb{\Sigma}, \pmb{\Pi}, \mathcal{G})$, and the associated block-diagram can be seen in Figure \ref{final_network_model}. Diffusively-coupled networks are of considerable interest in the control literature  \cite{Arcak2007,Burger2014,Montenbruck2015}, and include important examples such as neural networks \cite{Franci2011}, the Kuramoto model for oscillator synchronization \cite{Dorfler2014}, and traffic control models \cite{Bando1995}.

\begin{figure}[t!]
\centering
\includegraphics[width=3.7cm]{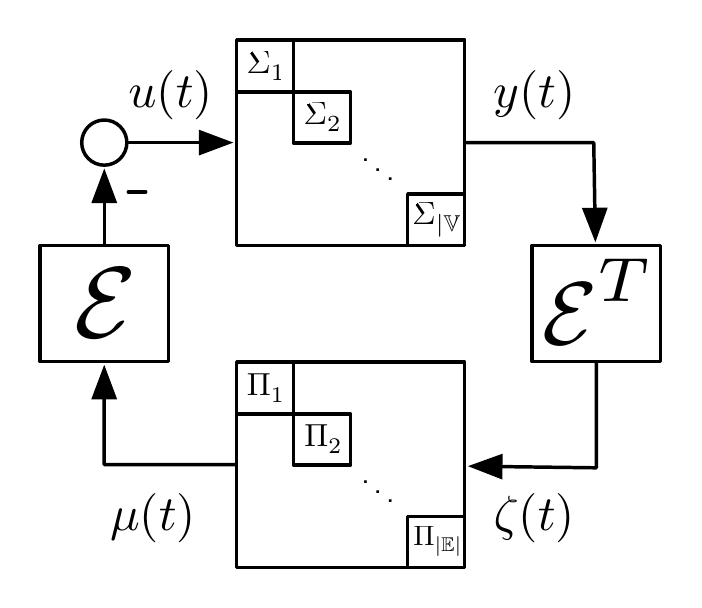}
\caption{A diffusively-coupled network.} 
\label{final_network_model} \vspace{-8pt}
\end{figure} 

The notion of MEIP allows us to connect between diffusively-coupled networks and network optimization theory.
\begin{thm}[\cite{Burger2014}]\label{thm_network_flow_problems}
 Consider the diffusively-coupled system $(\pmb{\Sigma}, \pmb{\Pi}, \mathcal{G})$. Suppose the agents are output-strictly MEIP and the controllers are MEIP, or vice versa. Let $K_i$ be the agents' integral functions, and let $\Gamma_e$ be the controllers' integral functions. We denote $\pmb{K}(\pmb{\rm u}) = \sum_{i\in \V} K_i(\rm u_i)$, $\pmb{\Gamma}(\pmb{\rm \mu}) = \sum_{e\in \EE} \Gamma_i(\rm \mu_i)$, and similarly for the Legendre transforms. Then there exist constant vectors $\pmb{\rm u}, \pmb{\rm y}, \pmb{\upzeta}, \pmb{\upmu}$ such the signals $\pmb{u}(t),\pmb{y}(t),\pmb{\zeta}(t),\pmb{\mu}(t)$ of $(\pmb{\Sigma}, \pmb{\Pi}, \mathcal{G})$ asymptotically converge to $\pmb{\rm u}, \pmb{\rm y}, \pmb{\upzeta}, \pmb{\upmu}$ correspondingly. Moreover, the steady-states $\pmb{\rm u}, \pmb{\rm y}, \pmb{\upzeta}$ and $\pmb{\upmu}$ are (dual) solutions of the following pair of convex optimization problems:
\begin{center}
\begin{tabular}{ c || c }
{\bf OFP} & {\bf OPP}\\
\hline
\parbox{1cm}{\begin{subequations}
\begin{alignat}{2}
\nonumber &\!\min_{\pmb{\rm u}, \pmb{\upmu}} &\qquad& \pmb{K}(\pmb{\rm u}) + \pmb{\Gamma}^\star(\pmb{\upmu})\\
\nonumber &{s.t.} &      & \pmb{\rm u} = - \E\pmb{\upmu} .
\end{alignat}
\end{subequations}} &  \parbox{1cm}{\begin{subequations}
\begin{alignat}{2}
\nonumber & \!\min_{\pmb{\rm y}, \pmb{\upzeta}} &\qquad& \pmb{K}^\star(\pmb{\rm y}) + \pmb{\Gamma}(\pmb{\upzeta})\\
\nonumber &{s.t.} &      & \E^\top\pmb{\rm y} = \pmb{\upzeta}
\end{alignat}
\end{subequations}}
\end{tabular}
\end{center} 
\end{thm}
These static optimization problems are known as the \emph{Optimal Flow Problem} 
(OFP) and the \emph{Optimal Potential Problem (OPP)}, and are dual to each 
other. These are classical problems in the mathematical field of network 
optimization, dealing with static optimization problems defined on graphs, and 
have been extensively studied by various researchers in fields as theoretical 
computer science and operations research \cite{Rockafellar1998}. However, this framework heavily relies on the passivity 
of the agents and controllers, and fails if any of the agents are not MEIP. As we'll see later, if the agents are not passive, the integral functions might be non-convex, or may not even exist.

\section{Motivation and Problem Formulation} \label{subsec.ProbForm}
Our end-goal is to extend the network optimization framework of Theorem \ref{thm_network_flow_problems} to agents which are not MEIP, but are rather EIPS. Unlike MEIP systems, EIPS systems need not have monotone steady-state relations. In some cases, this lack of monotonicity 
results in the non-convexity of the corresponding integral function \cite{Jain2018}, and in 
other cases, the steady-state  I/O relation is far enough from monotone that an 
integral function cannot even be defined. We give examples of this phenomenon 
in the following: 

\begin{example}[EI-OP($\rho$)]\label{exam_EI-OPS}
Consider a SISO system $\dot{x} = -x+\sqrt[3]{x}+u; y = \sqrt[3]{x}$. It is shown in \cite{Jain2018} that this system is EI-OP($\rho$) for all $\rho \le -1$, and its equilibrium-independent passivity index is $\rho = -1$. Moreover, the inverse steady-state I/O relation $\mathrm{u} = k^{-1}(\mathrm{y}) = \mathrm{y}^3 - \mathrm{y}$ is not monotone. Furthermore, it has an integral function $K^\star(\mathrm y) = \frac{1}{4}\mathrm y^4 - 
\frac{1}{2}\mathrm y^2$, which is non-convex due to the negative 
quadratic term.
\end{example}

\begin{example}[EI-IP($\nu$)]\label{exam_EI-IPS}
Consider the SISO system $\dot{x} = -\sqrt[3]{x} + u; y = x -u$. One can show similarly to Example \ref{exam_EI-OPS} that
this system is EI-IP($\nu$) for all $\nu \le -1$, and $\nu=-1$ is its equilibrium-independent passivity index. Moreover, the steady-state I/O relation $\mathrm{y} = k(\mathrm{u}) = \mathrm{u}^3 - \mathrm{u}$ is not monotone. Furthermore, it has an integral function $K(\mathrm u) = \frac{1}{4}\mathrm u^4 - \frac{1}{2}\mathrm u^2$, which is again non-convex due to the negative quadratic term.
\end{example}

\begin{example}[EI-IOP($\rho,\nu$)] \label{exam.InputOutputShortage}
Consider a SISO dynamical system $\Sigma$ given by
\begin{equation}\label{eq_former_system}
\Sigma:~\dot{x} = -\sqrt[3]{x} + 0.5x + 0.5u;~~~y = 0.5x - 0.5u,
\end{equation}
with input $u$ and output $y$. For any steady-state input-output pair $(\mathrm u,\mathrm y)$ and the corresponding state at equilibrium $\mathrm x = 2\mathrm y +\mathrm u$, we can consider the storage function $S(x) = \frac{1}{6} (x-\mathrm x)^2$. A simple calculation shows that:
\begin{align*}
\dot{S} \le (u-\mathrm u)(y-\mathrm y) + \frac{1}{3}(u-\mathrm u)^2 + \frac{2}{3}(y-\mathrm y)^2,
\end{align*}
meaning that the system is EI-IOP($\rho,\nu$) for $\rho = -2/3$ and $\nu = -1/3$. One can also easily verify that given an equilibrium state $\mathrm x$, the steady-state input $\mathrm u$ is given by $\mathrm u = 2\sqrt[3]{\mathrm x} - \mathrm x$ and that the steady-state output is $\mathrm y = \mathrm x - \sqrt[3]{\mathrm x}$. Defining $\sigma = -\sqrt[3]{\mathrm x}$, we see that the steady-state relation of the system is given by the planar curve $\mathrm u = 2\sigma-\sigma^3;\ \mathrm y = \sigma^3-\sigma$, parameterized by a variable $\sigma$, as shown in Figure~\ref{NoPrimalFunction}.  It is clear from Figure~\ref{NoPrimalFunction} that both steady-state I/O relation and its inverse are non-monotone. In fact, the 
steady-state input-output relation and its inverse are so far from monotone, no integral function exists for either of them. 

However, if we define a new input $\tilde{u}$ and a new output $\tilde{y}$ by $\tilde{u} = u+y, \tilde{y} = u+2y$, the resulting loop transformation gives the following system:
\begin{equation}\label{eq_latter_system}
\tilde{\Sigma}:~ \dot{x}=-\sqrt[3]{x}+\tilde{u};~~~\tilde{y} = x,
\end{equation} 
which has the steady-state input-output relation $k(\tilde{\mathrm u}) = \mathrm u^3$, which is maximally monotone. Moreover, the system \eqref{eq_latter_system} can be verified to be MEIP with storage function $S(x) = \frac{1}{2}(x-\mathrm x)^2$.
\end{example}
\begin{figure}[!t]
\begin{center}
\subfigure[The steady-state relation of \eqref{eq_former_system}.]{\scalebox{.31}{\includegraphics{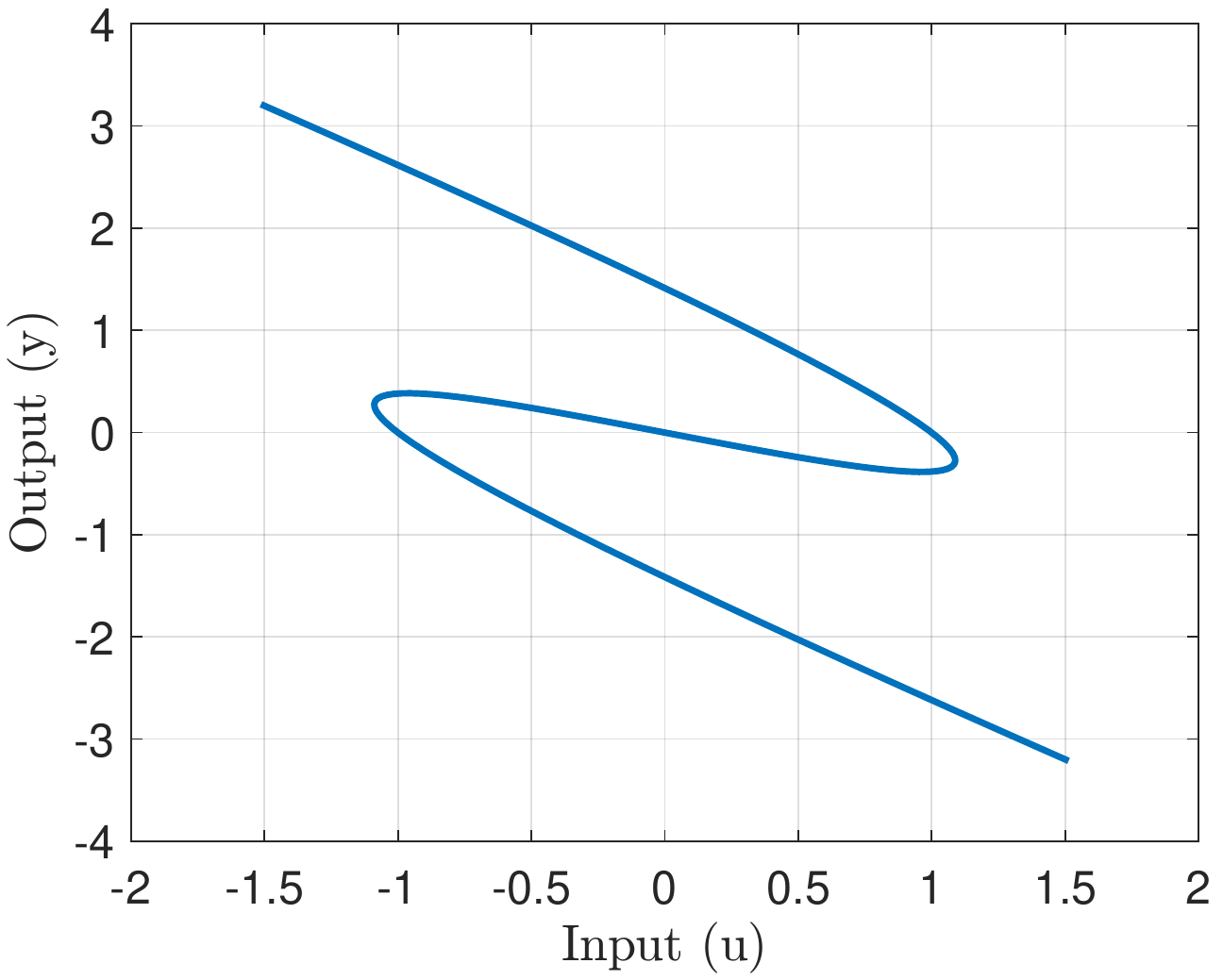}}}
\subfigure[The inverse relation of \eqref{eq_former_system}.]{\scalebox{.31}{\includegraphics{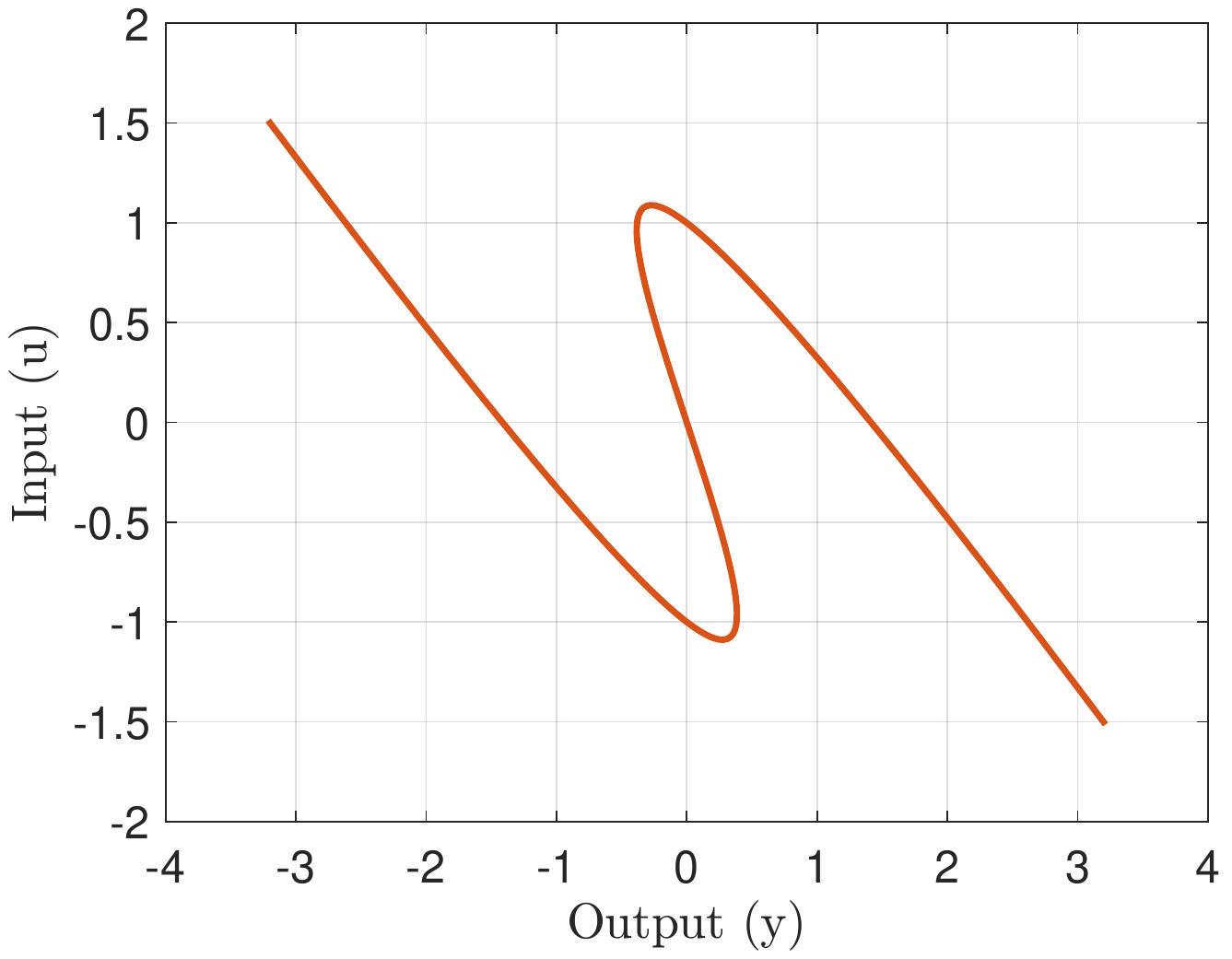}}}
\caption{Steady-state relations of the system in Example \ref{exam.InputOutputShortage}.} \label{NoPrimalFunction}
\end{center}
\vspace{-15pt}
\end{figure}

The above example shows that EIPS systems need not have integral 
functions, nor (maximally) monotone steady-state I/O relations. Thus, the network 
optimization framework of \cite{Burger2014} cannot even be defined for networks 
of EIPS agents. In \cite{Jain2018,Sharf2019c}, the network optimization 
framework failed due to the lack of convexity of the integral functions.  This 
was remedied by convexifying the resulting (non-convex) network optimization 
problems.  The interpretation (or implementation) of this convexification was a 
passivizing feedback term. We cannot follow this idea for EIPS 
systems when $\rho,\nu < 0$, as the network optimization framework is not even defined. 
Moreover, diffusely-coupled networks consisting of such systems might not be stable.
To overcome these shortcomings for EIPS systems, 
we investigate the existence of a loop transformation which results in 
monotonizing the steady-state I/O relation of the agents, as illustrated in the last part  of
Example~3. Thus, our goal in this paper is to find a monotonizing 
procedure for the steady-state I/O relation. We further show that the 
monotonizing procedure induces a passivizing plant transformation. 
For the rest of this paper, let $\Sigma$ be a EI-IOP($\rho,\nu$) system for known parameters $\rho,\nu$, and let $k$ be the corresponding steady-state relation.

\section{Monotonization of I/O Relations by Linear Transformations: A Geometric Approach} \label{Monotonization}
Our goal is to find a monotonizing transformation $T:(\mathrm u,\mathrm y)\mapsto(\tilde{\mathrm u},\tilde{\mathrm y})$ for $k$. We look for a linear transformation $T$ of the form $\left[\begin{smallmatrix} \tilde{\mathrm u} \\ \tilde{\mathrm y} \end{smallmatrix}\right] = T \left[\begin{smallmatrix} \mathrm u \\ \mathrm y \end{smallmatrix}\right]$.  Assuming the system is EI-IOP($\rho,\nu$) allows us to deduce information about the steady-state I/O relation:
\begin{proposition}\label{prop2}
Let $\Sigma$ be an EI-IOP($\rho,\nu$) system and let $k$ be its steady-state I/O relation. Then for any two points $(\mathrm u_1, \mathrm y_1),$ $(\mathrm u_2, \mathrm y_2)$ in $k$, the following inequality holds:
\begin{equation} \label{eq.SteadyStateQuadratic}
0 \le -\rho(\mathrm y_1 - \mathrm  y_2)^2 + (\mathrm u_1 -\mathrm u_2)(\mathrm y_1 - \mathrm y_2) - \nu(\mathrm u_1 - \mathrm u_2)^2.
\end{equation}
\end{proposition}
\begin{proof}

By definition of EI-IOP($\rho,\nu$), \eqref{IOSMEIP} holds for any steady-state $(\mathrm u,\mathrm y)$ and any trajectory $(u(t),x(t),y(t))$. Considering the steady-state $(\mathrm u_1,\mathrm y_1)$, we conclude that there exists a positive-definite storage function $S(x)$ such that the following inequality holds for all trajectories $(u(t),x(t),y(t))$:
\begin{align}\label{eq.Prop2Ineq}
\frac{dS}{dt} \leq -\rho (y - {\rm y_1})^2 -\nu (u - {\rm u_1})^2 + (y - {\rm y_1})(u - {\rm u_1}).
\end{align}
The steady-state input-output pair $(\mathrm u_2,\mathrm y_2)$ corresponds to some steady state $\mathrm x_2$, so that 
$(\mathrm u_2, \mathrm x_2, \mathrm y_2)$ is an (equilibrium) trajectory of the system. Plugging it into \eqref{eq.Prop2Ineq}, and noting that $\frac{d}{dt}S(\mathrm x_2) =0 $, we conclude that the inequality \eqref{eq.SteadyStateQuadratic} holds.
\end{proof}

Proposition~\ref{prop2} suggests the following definition:
\begin{definition}\label{def.PQI}
A \emph{projective quadratic inequality (PQI)} is an inequality with variables $\xi,\chi\in\R$ of the form
\begin{align} \label{eq.PQI}
0 \le a\xi^2 + b \xi\chi+c\chi^2,
\end{align}
for some numbers $a,b,c$, not all zero. The inequality is called \emph{non-trivial} if $b^2-4ac > 0$. The associated solution set of the PQI is the set of all points $(\xi,\chi)\in\R^2$ satisfying the inequality.
\end{definition}

By Definition \ref{def.PQI}, it is clear that \eqref{eq.SteadyStateQuadratic} 
is a PQI. Indeed, plugging $\xi = \mathrm u_1 - \mathrm u_2$ , $\chi = 
\mathrm y_1 - \mathrm y_2$ and choosing $a,b,c$ correctly verifies this. The 
demand $\rho\nu <\frac{1}{4}$ is equivalent to the non-triviality of the PQI. 
For example, monotonicity of the steady-state $k$ can 
be written as $0 \le (\mathrm u_1 - \mathrm u_2)(\mathrm y_1 - \mathrm y_2)$, 
which can be transformed to a PQI by choosing $a=c=0$ and $b=1$ in 
\eqref{eq.PQI}. Similarly, strict monotonicity can be modeled by taking $b=1$ 
and $a\le 0,c<0$.

As for transformations, the transformation $T = \left[\begin{smallmatrix} T_{11} & T_{12} \\ T_{21} & T_{22} \end{smallmatrix}\right]$ of the form 
$\left[\begin{smallmatrix} \tilde{\mathrm u} \\ \tilde{\mathrm y} \end{smallmatrix}\right] = T 
\left[\begin{smallmatrix} \mathrm u \\ \mathrm y \end{smallmatrix}\right]$ can be written as
$\tilde{\mathrm u} = T_{11} \mathrm u + T_{12} \mathrm y$ and 
$\tilde{\mathrm y} = T_{21} \mathrm u + T_{22} \mathrm y$. Plugging 
it inside \eqref{eq.SteadyStateQuadratic} gives another PQI. More precisely, 
if we let $F(\xi,\chi) = a\xi^2+b\xi\chi+c\chi^2$, and $T$ is a linear map, 
then $T$ maps the PQI $F(\xi,\chi) \ge 0$ to 
$F(T^{-1}(\tilde{\xi},\tilde{\chi})) \ge 0$. Our goal is to 
find a map $T$ transforming the inequality in Definition~\ref{def.PQI} to the PQI corresponding to monotonicity. Thus, 
we are compelled to consider the action of the group of linear transformations on the collection of PQIs.

Let $\mathcal{A}$ be the solution set of the original PQI. The connection between the original and transformed PQI described above shows that the solution set of the new PQI is $T(\mathcal{A})=\{T(\xi,\chi):\ (\xi,\chi)\in \mathcal{A}\}$. We can therefore study the effect of linear transformations on PQIs by studying their actions on the solution sets.  The action of the group of linear transformations on the collection of PQIs can be understood algebraically, but we use solution sets to understand it geometrically. We first give a geometric characterization of the solution sets.

\begin{note}
In this section, we abuse notation and identify the point $(\cos \theta,\sin \theta)$ on the unit circle $\mathbb{S}^1$ with the angle $\theta$ in some segment of length $2\pi$.
\end{note}

\begin{definition}
A \emph{symmetric section} $S$ on the unit circle $\mathbb{S}^1 \subseteq 
\mathbb{R}^2$ is the union of two closed disjoint sections that are 
opposite to each other, i.e., $S = B\cup (B+\pi)$ where $B$ is a closed section 
of angle $< \pi$. A \emph{symmetric double-cone} is defined as $A = \{\lambda 
s:\ \lambda >0, s\in \R\}$ for a symmetric section $S$.
\end{definition}

An example of a symmetric section and the associated symmetric double-cone can be seen in Figure \ref{fig.SymmetricSection}.

\begin{figure}[t!]
\centering
\includegraphics[width=3.2cm]{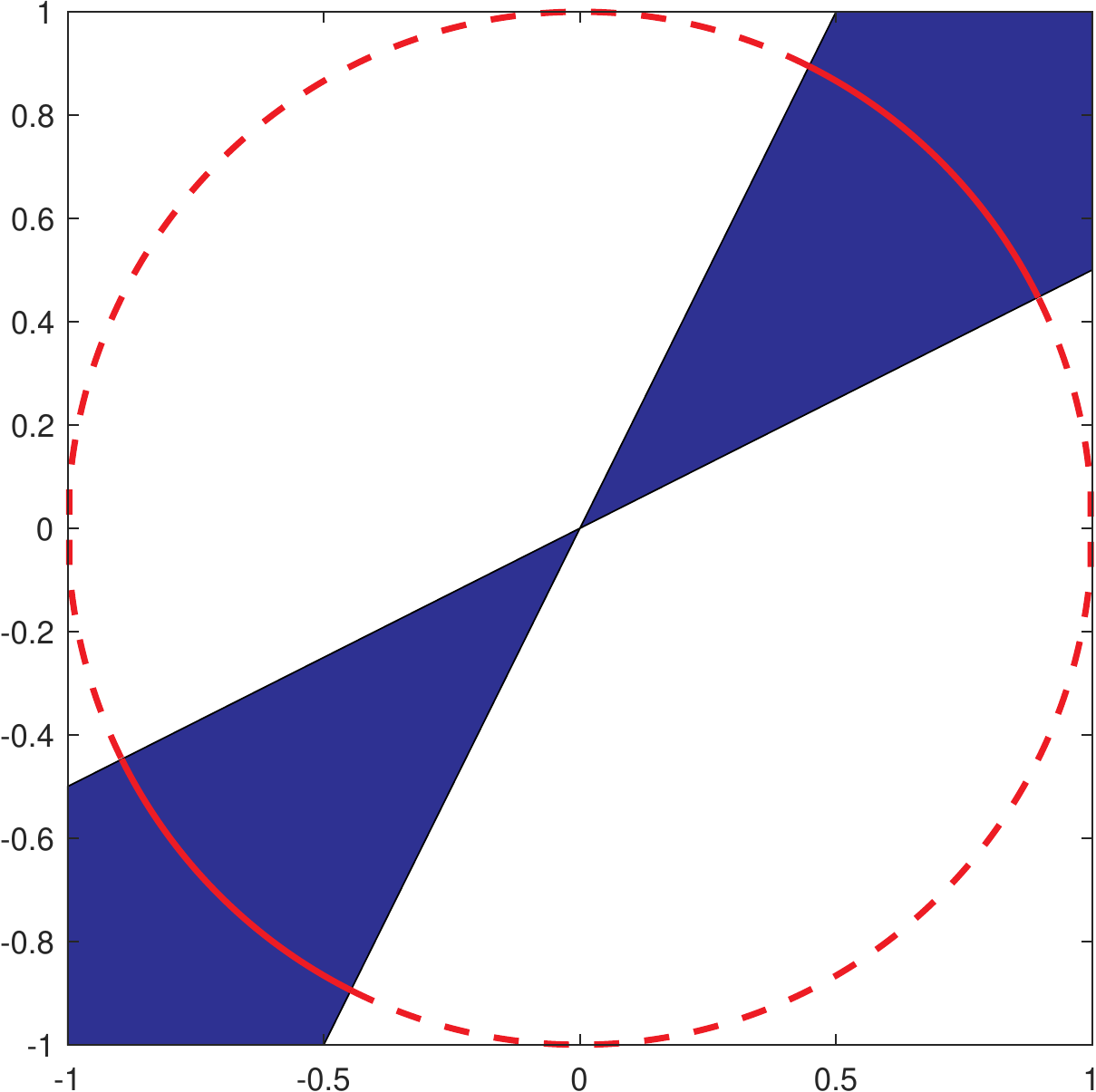}
\caption{A double cone (in blue), and the associated symmetric section $S$ (in solid red). The parts of $\mathbb{S}^1$ outside $S$ are presented by the dashed red line} 
\label{fig.SymmetricSection} \vspace{-5pt}
\end{figure} 

\begin{thm} \label{thm.GeometricApproach}
The solution set of any non-trivial PQI is a symmetric double-cone. Moreover, any symmetric double-cone is the solution set of some non-trivial PQI, which is unique up to a positive multiplicative constant.
\end{thm}
The proof of the theorem is available in the appendix. The theorem presents a geometric interpretation of the steady-state condition \eqref{eq.SteadyStateQuadratic}. The connection between cones and measures of passivity is best known for static systems through the notion of sector-bounded nonlinearities \cite{Khalil2002}. It was expanded to more general systems in \cite{Zames1966}, and later in \cite{McCourt2009}. We consider a different branch of this connection, focusing on the steady-state relation rather on trajectories. In turn, it allows us to have intuition when constructing monotonizing maps. In particular, we have the following result.
\begin{thm} \label{thm.MappingPQIs}
Let $(\xi_1,\chi_1)$, $(\xi_2,\chi_2)$ be two non-colinear solutions of $a_1\xi^2 + \xi\chi+c_1\chi^2=0$. Moreover, let $(\xi_3,\chi_3)$,$(\xi_4,\chi_4)$ be two non-colinear solutions of $a_2\xi^2 + \xi\chi+c_2\chi^2=0$. Define 
\begin{align}\label{eq.T1T2}
{T_1 = \left[\begin{matrix}\xi_3 & \xi_4 \\ \chi_3 & \chi_4\end{matrix}\right]\left[\begin{matrix}\xi_1 & \xi_2 \\ \chi_1 & \chi_2 \end{matrix}\right]^{-1}\hspace{-10pt},  T_2 = \left[\begin{matrix}\xi_3 & -\xi_4 \\ \chi_3 & -\chi_4\end{matrix}\right] \left[\begin{matrix}\xi_1 & \xi_2 \\ \chi_1 & \chi_2 \end{matrix}\right]^{-1}}.
\end{align}
Then one of $T_1,T_2$ transforms the PQI $a_1\xi^2 + \xi\chi+c_1\chi^2\ge0$ to the PQI $\tau a_2\xi^2 + \tau \xi\chi+\tau c_2\chi^2\ge 0$ for some $\tau>0$.
\end{thm}

The non-colinear solutions correspond to the straight lines forming the boundary of the symmetric double-cone, thus can be found geometrically. Moreover, as will be evident from the proof, knowing which one of $T_1$ and $T_2$ works is possible by checking the PQIs on $(\xi_1+\xi_2,\chi_1+\chi_2)$ and $(\xi_3+\xi_4,\chi_3+\chi_4)$. Namely, if exactly one of them satisfies the PQIs, then $T_2$ works, and otherwise $T_1$ works. We know present the proof of the theorem. 
\begin{proof}
Let $\mathcal{A}_1$ be the solution set of the first PQI, and let $\mathcal{A}_2$ be the solution set of the second PQI.
We show that either $T_1$ or $T_2$ maps $\mathcal{A}_1$ to $\mathcal{A}_2$.  We note that $T_1(\mathcal{A}_1)$ and $T_2(\mathcal{A}_1)$ are symmetric double-cones,
whose boundary is the image of the boundary of $\mathcal{A}_1$ under $T_1$ and $T_2$ respectively, i.e., they are the image of $\mathrm{span}\{(\xi_1,\chi_1)\}\cup \mathrm{span}\{(\xi_2,\chi_2)\}$ under $T_1,T_2$. 
We note that $T_1$ maps $(\xi_1,\chi_1)$,$(\xi_2,\chi_2)$ to $(\xi_3,\chi_3)$,$(\xi_4,\chi_4)$ correspondingly, and that $T_2$ maps $(\xi_1,\chi_1)$,$(\xi_2,\chi_2)$ to $(\xi_3,\chi_3)$,$(-\xi_4,-\chi_4)$ correspondingly. Thus, $\mathrm{span}\{(\xi_1,\chi_1)\}\cup \mathrm{span}\{(\xi_2,\chi_2)\}$ is mapped by $T_1$ and $T_2$ to $\mathrm{span}\{(\xi_3,\chi_3)\}\cup \mathrm{span}\{(\xi_4,\chi_4)\}$, so that $T_1(\mathcal{A}_1),T_2(\mathcal{A}_1)$ have the same boundary as $\mathcal{A}_2$. Since $T_1,T_2$ are homeomorphisms, they map interior points to interior points. Thus, it's enough to show that some point in the interior of $\mathcal{A}_1$ is mapped to a point in $\mathcal{A}_2$ either by $T_1$ or by 
$T_2$, or equivalently, that a point in the interior of $\R^2\setminus \mathcal{A}_1$ is mapped to a point in $\R^2\setminus \mathcal{A}_2$ either by $T_1$ or by $T_2$.

Consider the point $(\xi_1+\xi_2,\chi_1+\chi_2)$. By non-colinearity, this 
point cannot be on the boundary of $\mathcal{A}_1$, equal to 
$\mathrm{span}\{(\xi_1,\chi_1)\} \cup \mathrm{span}\{(\xi_2,\chi_2)\}$. Hence, 
it's either in the interior of $\mathcal{A}_1$ or in the interior of its 
complement. We assume the prior case, as the proof for the other is similar. 
The point $(\xi_1+\xi_2,\chi_1+\chi_2)$ is mapped to $(\xi_3\pm \xi_4,\ \chi_3 
\pm \chi_4)$ by $T_1$,$T_2$. By non-colinearity, these points do not lie on the 
boundary of $\mathcal{A}_2$.  Moreover, the line passing through them is 
parallel to $\mathrm{span}\{(\xi_4,\chi_4)\}$ which is part of the boundary of 
$\mathcal{A}_2$, and their average is $(\xi_3,\chi_3)$, which is on the 
boundary. Thus, one point is in the interior of $\mathcal{A}_2$, and one is in 
the interior of its complement. This completes the proof.
\end{proof}

\begin{example} \label{exam.4}
Consider the system $\Sigma$ studied in Example \ref{exam.InputOutputShortage}, in which the steady-state I/O relation was non-monotone. There, we saw that the system is EI-IOP($\rho,\nu$) with parameters $\rho = -2/3$ and $\nu = -1/3$. The corresponding PQI is $0 \le \frac{1}{3} \xi^2 + \xi\chi + \frac{2}{3}\chi^2$. We use Theorem \ref{thm.MappingPQIs} to find a monotonizing transformation. That is, we seek a transformation mapping the given PQI to the PQI defining monotonicity, $\xi\chi \ge 0$.
We take $(\xi_3,\chi_3)=(1,0)$ and $(\xi_4,\chi_4) = (0,1)$, as these are 
non-colinear solutions to $\xi\chi = 0$. For the original PQI, $0 = \frac{1}{3} \xi^2 + \xi\chi + \frac{2}{3}\chi^2$ can be 
rewritten as $\frac{1}{3}(\xi+\chi)(\xi+2\chi) = 0$, so we take 
$(\xi_1,\chi_1) = (2,-1)$ and $(\xi_2,\chi_2) = (-1,1)$. It's easy to check 
that $(\xi_1+\chi_1,\xi_2+\chi_2)=(1,0)$  satisfies the original PQI $0 \le 
\frac{1}{3} \xi^2 + \xi\chi + \frac{2}{3}\chi^2$, and that $(\xi_3+\chi_3,\xi_4+\chi_4)$ satisfies $\xi\eta \ge 0$ so the map $T_1$ defined in the Theorem~\ref{thm.MappingPQIs}, should monotonize the steady-state relation.
Plugging in $T_1$, we get $\left[\begin{smallmatrix} \xi \\ \chi \end{smallmatrix}\right] = T_1^{-1}\left[\begin{smallmatrix} \tilde{\xi} \\ \tilde{\chi} \end{smallmatrix}\right]$ for $T_1 = \left[\begin{smallmatrix}1 & 0 \\ 0 & 1\end{smallmatrix}\right]\left[\begin{smallmatrix} 2 & -1 \\ -1 & 1 \end{smallmatrix}\right]^{-1} = \left[\begin{smallmatrix} 1 & 1 \\ 1 & 2\end{smallmatrix}\right],$ so that $T_1^{-1} = \left[\begin{smallmatrix} 2 & -1 \\ -1 & 1 \end{smallmatrix}\right]$. Then,
\begin{align*}
0\le& \frac{1}{3} \xi^2 + \xi\chi + \frac{2}{3}\chi^2 \\ =&
\frac{1}{3}(2\tilde{\xi} - \tilde{\chi})^2 + (2\tilde{\xi} - \tilde{\chi})(-\tilde{\xi}+\tilde{\chi}) + \frac{2}{3}(-\tilde{\xi}+\tilde{\chi})^2 = \frac{1}{3} \tilde{\xi}\tilde{\chi},
\end{align*}
so the transformed PQI is $0 \le \tilde{\xi}\tilde{\chi}$, corresponding to monotonicity. To get the transformed steady-state relation, we recall that the steady-state relation of $\Sigma$ is given by the planar curve $\mathrm u = 2\sigma-\sigma^3;\ \mathrm y = \sigma^3-\sigma$, parameterized by a variable $\sigma$. The transformed relation is given by:
\begin{align*}
\left[\begin{matrix} \tilde{\mathrm u} \\ \tilde{\mathrm y} \end{matrix}\right] = T _1\left[\begin{matrix} \mathrm u \\ \mathrm y \end{matrix}\right] = \left[\begin{matrix} 1 & 1 \\ 1 & 2\end{matrix}\right]\left[\begin{matrix} 2\sigma-\sigma^3 \\ \sigma^3-\sigma \end{matrix}\right] = \left[\begin{matrix} \sigma \\ \sigma^3\end{matrix}\right],
\end{align*}
and can be modeled as $\tilde{\mathrm y} = \tilde{\mathrm u}^3$, which is a 
monotone relation.
\end{example}
Theorem \ref{thm.MappingPQIs} prescribes a monotonizing transformation for the relation 
$k$. Moreover, it prescribes a transformation forcing strict monotonicity, which can be viewed as the PQI $-\nu \xi^2 +\xi\chi \ge 0$ for $\nu \ge 0$, which are not both zero.

\section{From Monotonization to \\Passivation and Implementation} \label{Passivation}
Until now, we found a map $T:\R^2\to \R^2$, monotonizing the steady-state relation $k$. We claim $T$, in fact, transforms the agent $\Sigma$ into a system which is passive with respect to  all equilibria, by defining a new input and output as $\left[\begin{smallmatrix} \tilde{u} \\ \tilde{y} \end{smallmatrix}\right] = T \left[\begin{smallmatrix} u \\ y \end{smallmatrix}\right]$. 
\begin{proposition} \label{prop.MonInducedPassivation}
Let $\Sigma$ be EI-IOP($\rho,\nu$), and let $T$ be a map transforming the PQI $-\nu \xi^2 + \xi\chi -\rho \chi^2 \ge 0$ to $-\nu^{\prime}\xi^2 + \xi\chi - \rho^{\prime}\chi^2 \ge 0$ as in Theorem \ref{thm.MappingPQIs}. Consider 
the transformed system $\tilde{\Sigma}$ with input and output $\left[\begin{smallmatrix} \tilde{u} \\ 
\tilde{y} \end{smallmatrix}\right] = T\left[\begin{smallmatrix} u \\ y \end{smallmatrix}\right]$. Then $\tilde{\Sigma}$ is EI-IOP($\rho^\prime,\nu^\prime$). In particular, if $T$ monotonizes the relation $k$, it passivizes $\Sigma$.
\end{proposition}
\begin{proof}
The inequality \eqref{IOSMEIP} is the PQI $-\nu \xi^2 + \xi\chi - \rho \chi^2 \ge 0$, where we put $\xi = u(t) - \mathrm u$ and $\chi = y(t) - \mathrm y$ for a trajectory $(u(t),y(t))$ and a steady-state I/O pair $(\mathrm u,\mathrm y)$. The proposition follows by noting that $\left[\begin{smallmatrix} \xi \\ \chi \end{smallmatrix}\right] = T^{-1}\left[\begin{smallmatrix} \tilde{\xi} \\ \tilde{\chi} \end{smallmatrix}\right]$, satisfies the PQI $-\nu^{\prime}\xi^2 + \xi\chi - \rho^{\prime}\chi^2 \ge 0$, $\tilde{\xi} = \tilde{u}(t) - \tilde{\mathrm u}$ and $\tilde{\chi} = \tilde{y}(t) - \tilde{\mathrm y}$.
\end{proof}
Combining Theorem \ref{thm.MappingPQIs} and the discussion following it with Proposition \ref{prop.MonInducedPassivation} gives the following algorithm for passivation of EI-IOP($\rho,\nu$) systems with respect to all equilibria:

\begin{algorithm}
\caption{Passivation of an EI-IOP($\rho,\nu$) system}
{\bf Input} : A system $\Sigma$, and $\rho,\nu\in\mathbb{R}$ such that the system is EI-IOP($\rho,\nu$). Two more numbers $\rho^\prime,\nu^\prime$ such that $\rho^\prime \nu^\prime < 1/4$. \\
{\bf Output}: A transformation $T$, transforming the system $\Sigma$ to an EI-IOP($\rho^\prime,\nu^\prime$) system.
\label{alg.passivation}
\begin{algorithmic}[1] 
\State Find two pairs $(\xi_1,\chi_1), (\xi_2,\chi_2)$, which are non-colinear solutions of $-\nu \xi^2 + \xi\chi - \rho\chi^2 = 0$.
\State  Find two pairs $(\xi_3,\chi_3), (\xi_4,\chi_4)$, which are non-colinear solutions of $-\nu^\prime \xi^2 + \xi\chi - \rho^\prime\chi^2 = 0$.
\State Define $T_1, T_2$ as in \eqref{eq.T1T2}.
\State Define $\alpha_1 = -\nu (\xi_1+\xi_2)^2 + (\xi_1+\xi_2)(\chi_1+\chi_2) - \rho(\chi_1+\chi_2)^2$ and $\alpha_2 = -\nu^\prime (\xi_3+\xi_4)^2 + (\xi_3+\xi_4)(\chi_3+\chi_4) - \rho^\prime(\chi_3+\chi_4)^2$.
\If{$\alpha_1, \alpha_2$ are both non-positive or both non-negative,}
\State {\bf Return} $T_1$.
\Else
\State {\bf Return} $T_2$.
\EndIf 
\end{algorithmic}
\end{algorithm}

\begin{remark} \label{rem.ClassicalPassivity}
Proposition \ref{prop.MonInducedPassivation}, together with Section \ref{Monotonization}, prescribes a linear transformation passivizing the agent with respect to all equilibria. The same procedure can be applied to ``classical" passivity, in which one only looks at passivity with respect to a single equilibrium, as PQIs can be used to abstractify all dissipation inequalities. Our approach is entirely geometric and does not rely on algebraic manipulations.
\end{remark}

\begin{remark}
Note that if the transformation transforms $k$ to a \emph{strictly} monotone relation, the transformed system is \emph{strictly} passive.
\end{remark}

For the remainder of this section, we show that the I/O transformation can be easily 
implemented using standard control tools, namely gains, feedback and 
feed-through. We also connect the steady-state I/O relation $\lambda$ of the 
transformed system $\tilde{\Sigma}$ to $k$.
 
In this direction, take any linear map $T:\R^2\to \R^2$ of the form $T = \left[\begin{smallmatrix} a & b \\  
c & d \end{smallmatrix}\right],$ where we assume that $\det(T) \neq 0$. It defines the 
plant transformation of the form
$
\left[\begin{smallmatrix}
    \tilde{u} \\
    \tilde{y} 
\end{smallmatrix}\right]
= T
\left[\begin{smallmatrix}
    {u} \\
    {y} 
\end{smallmatrix}\right].
$
For simplicity of presentation, we assume that $a\neq 0$.\footnote{We note that 
by switching the names of $(\xi_3,\chi_3)$ and $(\xi_4,\chi_4)$ in Theorem 
\ref{thm.MappingPQIs}, we switch the two columns of $T$. Thus we can always 
assume that $a\neq 0$, as $a=b=0$ cannot hold due to the determinant 
condition.} We note $T$ can be written as a product of elementary 
matrices, and the effect of each elementary matrix on $\Sigma$ can be easily 
understood. By applying the elementary transformations sequentially, the effect 
of their product, $T$, can be realized. Table~\ref{table} summarizes the elementary transformations and their effect 
on the system $\Sigma$. Following Table~\ref{table}, $T$ is written as 
\begin{gather}\label{eq_transformation}
T = \left[\begin{matrix}
    a & b \\
    c & d 
\end{matrix}\right]
= \underbrace{\left[\begin{matrix}
    \delta_D & 0 \\
    0 & 1 
\end{matrix}\right]}_{L_D}\underbrace{\left[\begin{matrix}
   1 & 0 \\
    \delta_C & 1 
\end{matrix}\right]}_{L_C}\underbrace{\left[\begin{matrix}
    1 & 0 \\
    0 & \delta_B 
\end{matrix}\right]}_{L_B}\underbrace{\left[\begin{matrix}
    1 & \delta_A \\
    0 & 1 
\end{matrix}\right]}_{L_A},
\end{gather}
with $\delta_A = b/a, \delta_B =d-\frac{b}{a}c,  \delta_C = c$ and $\delta_D = a$. The product of these matrices can be seen as the sequential transformation from the original system ${\Sigma}$, which can be understood as a loop-transformation, illustrated in Figure \ref{Feedback Passivation}.

\begin{remark}
Writing $T=L_DL_CL_BL_A$ allows us to give a closed form description of the transformed system. Suppose the original system is given by $\dot{x} = f(x,u);\ y=h(x)$. Applying $L_A$ gives a new input $v$, and the transformed system $\dot{x} = f(x,v-\delta_A h(x)); y=h(x)$. Applying $L_B$ on this system gives $\dot{x} = f(x,v-\delta_A h(x)); y=\delta_B h(x)$. Applying $L_C$ then gives $\dot{x} = f(x,v-\delta_A h(x)); y=\delta_B h(x)+\delta_C v$, and applying $L_D$ finally gives $\dot{x} = f(x,\delta_D v-\delta_A h(x)); y=\delta_B h(x)+\delta_C \delta_D v$.
\end{remark}

\begin{figure}[t!]
\centering
\includegraphics[width=8cm]{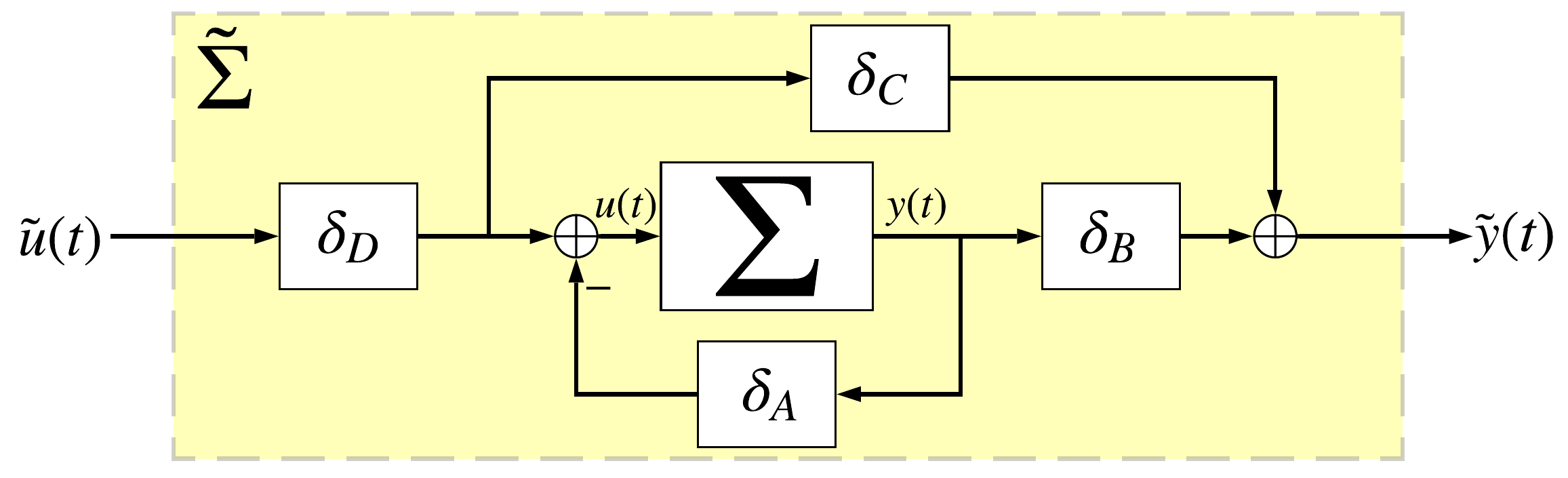}
\caption[The transformed system $\tilde{\Sigma}$ after the linear transformation $T$.]{The transformed system $\tilde{\Sigma}$ after the linear transformation $T$. If $T=\left[\protect\begin{smallmatrix} a & b \\ c & d \protect\end{smallmatrix}\right]$, then $\delta_A = b/a, \delta_B =d-\frac{b}{a}c,  \delta_C = c$ and $\delta_D = a$.} 
\label{Feedback Passivation}\vspace{-5pt}
\end{figure}

\begin{table*}[t]
\caption{Elementary matrices and their realizations}
\begin{center}
\begin{tabular}{ |>{\centering\arraybackslash}p{2.4cm}|>{\centering\arraybackslash}p{2.7cm}|>{\centering\arraybackslash}p{3.7cm}|>{\centering\arraybackslash}p{1.7cm}|>{\centering\arraybackslash}p{3.8cm}| } 
 \hline
 Elementary Transformation & Relation between I/O of $\Sigma$ and $\tilde{\Sigma}$ & Effect on Steady-State Relations& Realization & Effect on Integral Functions\\ 
 \hline
 \hline
 $L_A = \begin{bmatrix} 1 & \delta_A \\ 0 & 1 \end{bmatrix}$ & $\tilde{u} = u + \delta_Ay$ ~~~~~~$\tilde{y} = y$ & ${\lambda}_A^{-1}(\tilde{\mathrm y}) = {k}^{-1}(\tilde{\mathrm y}) + \delta_A\tilde{\mathrm y}$ & output-feedback & $\Lambda^\star(\mathrm y) = K^\star(\mathrm y) + \frac{1}{2}\delta_A \mathrm y^2$ \\ 
 \hline 
$L_B = \begin{bmatrix} 1 & 0 \\ 0 & \delta_B \end{bmatrix}$ & $\tilde{u} = u$ ~~~~~~~~~~ $\tilde{y} = \delta_B y$ & ${\lambda}_B({\mathrm u}) = \delta_B {k}({\mathrm u})$ or ${\lambda}_B^{-1}(\tilde{\mathrm y}) = {k}^{-1}(\frac{1}{\delta_B}\tilde{\mathrm y})$ & post-gain & $\Lambda^\star(\mathrm y) = \frac{1}{\delta_B} K^\star(\frac{1}{\delta_B}\mathrm y)$ or $\Lambda(\mathrm u) = {\delta_B}K(\mathrm u)$\\ 
 \hline
  $L_C = \begin{bmatrix} 1 & 0 \\ \delta_C & 1 \end{bmatrix}$ & $\tilde{u} = u$ ~~~ $\tilde{y} = y + \delta_Cu$ & ${\lambda}_C(\tilde{\mathrm u}) = {k}(\tilde{\mathrm u}) + \delta_C\tilde{\mathrm u}$ & input-feedthrough & $\Lambda(\mathrm u) = K(\mathrm u) + \frac{1}{2}\delta_C \mathrm u^2$\\ 
 \hline
 $L_D = \begin{bmatrix} \delta_D & 0 \\ 0 & 1 \end{bmatrix}$ & $\tilde{u} = \delta_Du$ ~~~~~~~~~~$\tilde{y} = y$ & ${\lambda}_D^{-1}(\mathrm y) = \delta_D{k}^{-1}(\mathrm y)$ or ${\lambda}_D(\tilde{\rm u}) = {k}(\frac{1}{\delta_D}\tilde{\rm u})$ & pre-gain & $\Lambda^\star(\mathrm y) = \delta_D K^\star(\mathrm y)$ or $\Lambda(\mathrm u) = \frac{1}{\delta_D}K(\frac{1}{\delta_D}\mathrm u)$ \\ 
 \hline
\end{tabular}
\label{table}
\end{center}
\end{table*}

\begin{proposition}\label{prop_connection_between_relations}
Let $k$ and $\lambda$ be the steady-state I/O relations of $\Sigma$ and $\tilde{\Sigma}$, respectively, where $\tilde{\Sigma}$ is the result of applying the transformation $T$ in \eqref{eq_transformation} on $\Sigma$, where $\delta_A = b/a, \delta_B =d-\frac{b}{a}c,  \delta_C = c$ and $\delta_D = a$. Assume that $\kappa_{1}$ is the steady-state I/O relation for the system $\Sigma_{1}: u_{1} \mapsto y_{1}$, obtained after the transformation $L_A = \left[\begin{smallmatrix} 1 & \delta_A \\ 0 & 1 \end{smallmatrix}\right]$ on the original system $\Sigma$. Then, the relation between $\lambda$ and $k$ is given by
\begin{equation} \label{eq.prop4_1}
\lambda(\tilde{\rm u}) = \left(d-\frac{b}{a}c\right)\kappa_{1}\left(\frac{1}{a}\tilde{\rm u}\right) + \frac{c}{a}\tilde{\rm u},
\end{equation}
where the inverse of $\kappa_{1}$ is 
\begin{equation} \label{eq.prop4_2}
(\kappa_{1})^{-1}({\rm y}_{1}) = k^{-1}({\rm y}_{1}) + \frac{b}{a}{\rm y}_{1}.
\end{equation}
\end{proposition}

\begin{proof}
Denote the steady-state I/O relations after the first, second, and third elementary matrix transformations, sequentially in \eqref{eq_transformation}, as $\kappa_{1}, \kappa_{2}, \kappa_{3}$, corresponding to the steady-state I/O pairs $({\rm u}_{1}, {\rm y}_{1})$, $({\rm u}_{2}, {\rm y}_{2})$ and $({\rm u}_{3}, {\rm y}_{3})$. The transformation 
$$\left[\begin{matrix}
{\rm u}_{1} \\
{\rm y}_{1}
\end{matrix}\right]
= L_A\left[\begin{matrix}
{\rm u} \\
{\rm y}
\end{matrix}\right]=
\left[\begin{matrix}
    1 & b/a \\
    0 & 1 
\end{matrix}\right]
\left[\begin{matrix}
{\rm u} \\
{\rm y}
\end{matrix}\right],$$
has the steady-state inverse I/O relation $\kappa_{1}^{-1}(\mathrm{y}_{1}) = k^{-1}(\mathrm{y}_{1}) + \frac{b}{a}{\mathrm{y}_{1}}$. 
The second transformation 
$$\left[\begin{matrix}
{\rm u}_{2} \\
{\rm y}_{2}
\end{matrix}\right]
=L_B\left[\begin{matrix}
{\rm u_1} \\
{\rm y_1}
\end{matrix}\right]=
\left[\begin{matrix}
    1 & 0 \\
    0 & d-bc/a 
\end{matrix}\right]
\left[\begin{matrix}
{\rm u}_{1} \\
{\rm y}_{1}
\end{matrix}\right],$$
has the steady-state I/O relation $\kappa_{2}(\mathrm{u}_{2}) = (d - \frac{b}{a}c)\kappa_1(\mathrm{u}_{2})$. The third transformation 
$$\left[\begin{matrix}
{\rm u}_{3} \\
{\rm y}_{3}
\end{matrix}\right]
=L_C\left[\begin{matrix}
{\rm u_2} \\
{\rm y_2}
\end{matrix}\right]=
\left[\begin{matrix}
    1 & 0 \\
    c & 1 
\end{matrix}\right]
\left[\begin{matrix}
{\rm u}_{2} \\
{\rm y}_{2}
\end{matrix}\right],$$
has steady-state I/O relation $\kappa_{3}(\mathrm{u}_{3}) = \kappa_{2}(\mathrm{u}_{3}) + c \mathrm{u}_{3}$. Finally, 
$$\left[\begin{matrix}
\tilde{\rm u} \\
\tilde{\rm y}
\end{matrix}\right]
= L_D\left[\begin{matrix}
{\rm u_3} \\
{\rm y_3}
\end{matrix}\right]=
\left[\begin{matrix}
    a & 0 \\
    0 & 1 
\end{matrix}\right]
\left[\begin{matrix}
{\rm u}_{3} \\
{\rm y}_{3}
\end{matrix}\right],$$
has the steady-state I/O relation $\lambda$ of $\tilde{\Sigma}$, and $\lambda(\tilde{\mathrm u}) = \kappa_{3}(\frac{1}{a}\tilde{\mathrm u})$. Substituting back for $\kappa_{3}$ and for $\kappa_{2}$, we get the result.  
\end{proof}

\begin{example}
Consider the system in Examples~\ref{exam.InputOutputShortage} and \ref{exam.4}. The steady-state 
I/O relation $\lambda$ of $\tilde{\Sigma}$ consists of all pairs $(\tilde{\rm u}, 
\tilde{\rm u}^3)$. We use Proposition~\ref{prop_connection_between_relations} 
to verify this result. According to 
Proposition~\ref{prop_connection_between_relations}, for the given matrix 
transformation $T = \left[\begin{smallmatrix} 1 & 1 \\ 1 & 2 \end{smallmatrix}\right]$, $\lambda$ is 
given by $\lambda(\tilde{\rm u}) = \kappa_1(\tilde{\rm u}) + \tilde{\rm u}$. 
After the first transformation $L_A = \left[\begin{smallmatrix} 1 & 1 \\ 0 & 1 
\end{smallmatrix}\right]$, the steady-state I/O pairs of the system $\Sigma_1$ are ${\rm 
u}_1 = {\rm u} + {\rm y}$, and ${\rm y}_1 = {\rm y}$. Substituting $\mathrm u = 
2\sigma-\sigma^3$, and $\mathrm y=\sigma^3-\sigma$ as obtained in 
Example~\ref{exam.InputOutputShortage} yields ${\rm u}_1 = \sigma$ and hence 
$\kappa_1({\rm u}_1) = {\rm y}_1 = {\rm u}_1^3 - {\rm u}_1$. This implies that 
$\kappa_1(\tilde{\rm u}_1) = {\rm u}_1^3 - {\rm u}_1$, which on substitution 
yields $\lambda(\tilde{\rm u}) = \tilde{\rm u}^3$, as expected.  
\end{example}


As discussed above, in some cases, i.e., when $\rho,\nu \geq 0$, we know the original system possesses integral functions. We can integrate  \eqref{eq.prop4_1} and \eqref{eq.prop4_2}, obtaining a connection between the original and the transformed integral functions. For example, integrating the steady-state equation for output-feedback $\lambda^{-1}(\tilde{\mathrm y}) = k^{-1}(\tilde{\rm y}) + \delta\tilde{\rm y}$ results in $K^{\star}(\tilde{\mathrm y}) = \Lambda^{\star}(\tilde{\rm y}) + \frac{\delta}{2} \tilde{\rm y}^2$, where $K^\star,\Lambda^\star$ are the integral functions of $k^{-1},\lambda^{-1}$ respectively. Similarly, input-feedthrough corresponds to a quadratic term added to the integral function $K$ of $k$, and pre- and post-gain correspond to scaling the integral function. These connections are summarized in Table \ref{table}.

\begin{example} \label{exam.IntegralFunctionsConnection}
Consider Example \ref{exam_EI-OPS}. The steady-state input-output relation for the system is $\mathrm u = k^{-1}(\mathrm y) = \mathrm y^3 - \mathrm y$, so the corresponding integral function is $K^\star(\mathrm y) = 
\frac{1}{4} \mathrm y^4 - \frac{1}{2}\mathrm y^2$. Consider the transformation $T = \left[\begin{smallmatrix} 1 & 1 \\ 0 & 1 \end{smallmatrix}\right]$, or equivalently $\tilde{u} = u+y = u + \sqrt[3]{x}, \tilde{y} = y$, so $u 
= -\sqrt[3]{x} + \tilde{u}$. The transformed system $\tilde{\Sigma}$ has the state-space model $\dot{x} = -x+\tilde{u}, \tilde{y} = 
\sqrt[3]{x}$, which has a steady-state I/O relation of $\tilde{\mathrm 
u} = \lambda^{-1}(\tilde{\mathrm y}) = \tilde{\mathrm y}^3$, and
corresponding integral function is $\Lambda^\star(\tilde{\mathrm y}) = 
\frac{1}{4}\tilde{\mathrm y}^4$. It is evident that $\Lambda^\star(\mathrm y) = 
K^\star(\mathrm y) + \frac{1}{2}\mathrm y^2$, as forecasted by Table \ref{table}. 
\end{example}

The passivation results achieved up to now assumed that the system at hand is EIPS. In the next section, we connect this property to having a finite $\mathcal{L}_2$-gain, showing our results extend \cite{Xia2018}.

\section{Finite $\mathcal{L}_2$-gain and Input-Output Passivity} \label{Finite L2 Gain}
This section establishes a connection between the notion of input-output $(\rho,\nu)$-passivity and the finite $\mathcal{L}_2$-gain property, and compares our results with the existing literature. We further explore these connections for the special case of linear and time-invariant systems and draw some important conclusions.

\subsection{Finite $\mathcal{L}_2$-gain and Input-Output $(\rho,\nu)$-Passivity}

We begin with by recalling the definition of systems with finite $\mathcal{L}_2$-gain.
\begin{definition}
The system $\Sigma:u \mapsto y$ has finite-$\mathcal{L}_2$-gain with respect to the steady-state I/O pair $(\mathrm{u,y})$ if there exists some $\beta > 0$ and a storage function $S$ such that:
\begin{align} \label{L2Ineq}
\dot{S} \le -(y-\mathrm y)^\top(y-\mathrm y) + \beta^2(u - \mathrm u)^\top(u-\mathrm u).
\end{align}
The smallest number $\beta$ satisfying the dissipation inequality is called the $\mathcal{L}_2$-gain of the system $\Sigma$.
\end{definition}
The notion of systems with a finite $\mathcal{L}_2$-gain can also be understood using the operator norm, namely, a system $\Sigma: u \mapsto y$ has a finite $\mathcal{L}_2$-gain if and only if its induced operator norm $\sup_{u\neq 0} \frac{\|\Sigma(u)\|}{\|u\|}$ is finite. In that case, the $\mathcal{L}_2$-gain is equal to the operator norm \cite{Khalil2002}. We now show that any system with a finite $\mathcal{L}_2$-gain is actually input passive-short, and thus included in the collection of input-output $(\rho,\nu)$-passive systems.
\begin{thm}\label{thm_IOPS_Finite_Gain}
Let $\Sigma: u \mapsto y$ be any finite $\mathcal{L}_2$-gain system with respect to the steady-state input-output pair $(\mathrm{u,y})$ with gain $\beta$. Then $\Sigma$ is input $\nu$-passive with respect to $(\mathrm{u,y})$, in the sense of Definition~\ref{defn_passive_short}, where $\nu \le -\left(\beta^2 + \frac{1}{4}\right).$
\end{thm}

\begin{proof}
Let $S(x)$ be the storage function corresponding to the finite $\mathcal{L}_2$-gain system $\Sigma$. By assumption, we know that for any trajectory $(u(t),x(t),y(t))$, the following inequality holds:
\begin{align*}
\frac{dS}{dt}(x) \le -\|y(t)-\mathrm y\|^2 + \beta^2 \|u(t) - \mathrm u\|^2.
\end{align*}
We note that $\|y(t)-\mathrm y+0.5(u(t) - \mathrm u)\|^2 \ge 0$, implying that $-\|y(t)-\mathrm y\|^2 \le (u(t)-\mathrm u)^\top (y(t)-\mathrm y)+0.25\|u(t)-\mathrm u\|^2$. Thus, we conclude that
\begin{align*}
\frac{dS}{dt}(x) &\le -\|y(t)-\mathrm y\|^2 + \beta^2 \|u(t)-\mathrm u\|^2 \\
&\le (u(t)-\mathrm u)^\top (y(t)-\mathrm y) + \left(\beta^2 +\frac{1}{4}\right)\|u(t)-\mathrm u\|^2,
\end{align*}
implying that $\Sigma$ is input $\nu$-passive with respect to $(\mathrm{u,y})$. This concludes the proof of the claim.
\end{proof} 
\begin{remark}
One can easily check that the above result is not true in the opposite direction, that is, if the system $\Sigma$ is EI-IP($\nu$) for some $\nu$, it does not necessarily have a finite $\mathcal{L}_2$-gain. Thus, the consideration of EIPS system is more general when compared to finite-$\mathcal{L}_2$-gain systems as in \cite{Xia2018}. Subsection \ref{exam.LTI_RelDeg} gives an example of a system which is EIPS
but neither input passive-short, output passive-short, nor does it have a finite $\mathcal{L}_2$-gain.
\end{remark}

\begin{remark}
Systems with a finite $\mathcal{L}_2$-gain have in important use in approximation theory. In many examples, we do not have an exact model for a system $\Sigma$, but instead we are given a model for an approximate model $\Sigma_0$ and a bound on the approximation error $\Sigma-\Sigma_0$, usually in terms of its $\mathcal{L}_2$-gain. In this case, proving that $\Sigma_0$ satisfies some dissipation inequality might be easy, but trying to directly find such an inequality satisfied by $\Sigma$ can be an arduous task. However, \cite{Xia2017} describes a method to prove a dissipation inequality for $\Sigma$ using a dissipation inequality for $\Sigma_0$ and an estimate on the $\mathcal{L}_2$-gain of the approximation error $\Sigma-\Sigma_0$. The achieved dissipation inequality might be very conservative, but we can still apply Algorithm \ref{alg.passivation}, as it does not need the exact passivity indices, but only some bound on them. In particular, the presented approach works even when we are only given an approximation of the true system.
\end{remark}

\subsection{Equilibrium-Independent Passive Shortage and Linear and Time-Invariant Systems}
This subsection drives an important result for the linear and time-invariant systems (LTI) relating their transfer function and passivity indices. LTI systems are of special interest for equilibrium-independent notions of passivity, as they are equivalent to the corresponding classical notions of passivity with respect to the steady-state pair $(0,0)$. For example, the proof of Theorem \ref{eq.thm8} below shows that an LTI system is EI-IOP($\rho,\nu$) if and only if it is input-output $(\rho,\nu)$-passive with respect to the steady-state $(0,0)$, if and only if the associated transfer function is input-output $(\rho,\nu)$-passive. This theorem shows that a vast class of LTI systems are EIPS, and calculates a bound on their passivity indices.

\begin{thm} \label{eq.thm8}
Let $\Sigma$ be a linear time-invariant system, and let $G(s) = \frac{p(s)}{q(s)}$ be the corresponding transfer function, where we assume that $p(s),q(s)$ are coprime and that $\deg p\le \deg q$. Suppose that there exists some $\lambda \in \mathbb{R}$ such that $q(s) + \lambda p(s)$ is a stable polynomial, i.e., all of its roots are in the open left-half plane, with degree equal to $\deg q$.
Define 
\begin{align}\label{transformed_l2gain}
\mu = \sup_{\omega \in \mathbb{R}} \left|\frac{p(j\omega)}{q(j\omega) + \lambda p(j\omega)}\right|^2+\frac{1}{4}.
\end{align} Then $\Sigma$ is EI-IOP($\rho,\nu$), where $\rho = -\frac{\lambda(1+\lambda\mu)}{1+2\lambda\mu}$ and $\nu = - \frac{\mu}{1+2\lambda\mu}$.
\end{thm}

\begin{proof}
Let $(\mathrm u,\mathrm y)$ be a steady-state input-output pair of the system, so that $\mathrm y = G(0)\mathrm u$. The system $\Sigma$ is input-output $(\rho,\nu)$-passive with respect to $(\mathrm u,\mathrm y)$ if and only if the corresponding operator $\Sigma_{\rm shifted}: \bar{u} \mapsto \bar{y}$ is input-output $(\rho,\nu)$-passive, where $\bar{u} = u-\mathrm u$ and $\bar{y} = y - \mathrm{y}$. If we let $(A,B,C,D)$ be a state-space representation of $G(s)$, then the operator $\Sigma_{\rm shifted}$ has the following (shifted) state-space realization:
\begin{align*}
\dot{x} = Ax + B(u-\mathrm u) ; ~~y = Cx + D(u - \mathrm u) + \mathrm y.
\end{align*}
Recalling that $G(0) = -CA^{-1}B + D$ and $\mathrm y = G(0)\mathrm u$, we conclude $\Sigma_{\rm shifted}$ is also linear and time-invariant, and its transfer function is equal to $G(s)$.

We now let $\tilde\Sigma_{\rm shifted}$ be the interconnection of the system $\Sigma_{\rm shifted}$ with a negative output-feedback with gain equal to $\lambda$. It is straightforward to show that $\tilde \Sigma_{\rm shifted}$ is also an LTI system, and its transfer function is $\tilde{G}(s)  = \frac{p(s)}{q(s)+\lambda p(s)}$. By assumption, all poles of the denominator are in the open left-half plane, and the degree of the numerator is bounded by the degree of the denominator. Thus, $\tilde \Sigma_{\rm shifted}$ has a finite $\mathcal{L}_2$-gain with respect to the origin, equal to $\kappa = \sup_{\omega \in \mathbb{R}} |\tilde G(j\omega)|$ \cite{Khalil2002}. We denote the input of the new system by $\tilde{\bar{u}} = \bar{u} - \lambda \bar{y}$.

Let $S(x)$ be the storage function corresponding to $\tilde \Sigma_{\rm shifted}$. We take an arbitrary trajectory $(\bar{u}(t),x(t),\bar{y}(t))$ of $\Sigma$ and consider the corresponding trajectory $(\tilde{\bar{u}}(t),x(t),\bar{y}(t))$ for $\tilde{\Sigma}_{\rm shifted}$, where $\bar{u}(t) = \tilde{\bar{u}}(t) - \lambda \bar{y}(t)$. As $\tilde \Sigma_{\rm shifted}$ has a finite $\mathcal{L}_2$-gain equal to $\kappa$, the following inequality holds:
\begin{align}\label{eq.Thm8Eq1}
\dot{S}(x) \le -\bar{y}(t)^2 + \kappa^2 \tilde{\bar{u}}(t)^2.
\end{align}
We note that $(\bar{y}(t)+0.5\tilde{\bar{u}}(t))^2 \ge 0$, so $-\bar{y}(t)^2 \le \tilde{\bar{u}}(t)\bar{y}(t) + 0.25\tilde{\bar{u}}(t)^2$. By plugging it into \eqref{eq.Thm8Eq1}, and recalling that $\kappa ^2 + 0.25 = \mu$ (by \eqref{transformed_l2gain}), we conclude that:
\begin{align*}
\dot{S}(x) \le &~ \tilde{\bar{u}}\bar{y} + \mu \tilde{\bar{u}}^2 = (\bar{u} + \lambda \bar{y})\bar{y} + \mu (\bar{u} +\lambda \bar{y})^2 \\= &  ~
\bar{u}\bar{y} + \lambda \bar{y}^2 + \mu \bar{u}^2 + 2\lambda\mu \bar{u}\bar{y} + \mu\lambda^2 \bar{y} \\=& ~ (1+2\mu\lambda)\bar{u}{y} + \mu \bar{u}^2 +(\lambda + \mu\lambda^2) \bar{y}^2 \\=&~  (1+2\mu\lambda)(\bar{u}\bar{y} - \nu \bar{u}^2 - \rho \bar{y}^2).
\end{align*}
Choosing the storage function $R(x) = S(x)/(1+2\mu\lambda)$, as well as recalling that $\bar{u} = u - \mathrm u$ and $\bar{y} = y - \mathrm y$, shows that $\Sigma$ is input-output ($\rho,\nu$)-passive with respect to the input-output steady-state pair $(\mathrm u,\mathrm y)$. As the steady-state pair was arbitrary, we conclude $\Sigma$ is EI-IOP($\rho,\nu$) with the passivity indices as defined in the statement of theorem.
\end{proof}

Recall that in Section~\ref{Passivation}, we presented a method of taking an EIPS system and transforming it to another system which is passive with respect to all equilibria. In the following section, we deal with the last ingredient missing for MEIP, namely maximality of the acquired monotone relation.

\section{Maximality of Input-Output Relations and the Network Optimization Framework}\label{Cursive}
As we saw, the map $T$ monotonizes the steady-state relation $k$, i.e., the steady-state input-output relation $\lambda$ of 
the transformed agent $\tilde{\Sigma}$ is monotone. However, it does not guarantee that $\lambda$ is \emph{maximally monotone}, which is essential for applying Theorem \ref{thm_network_flow_problems}. In this section, we explore a possible way to assure that $\lambda$ is maximally monotone, under which we prove a version of Theorem \ref{thm_network_flow_problems} for EIPS systems.

\begin{definition}[Cursive Relations] \label{defn_cursive_relation}
A set $A \subset \mathbb{R}^2$ is called \emph{cursive} if there exists a curve \footnote{A curve is a continuous map from a (possibly infinite) interval in $\mathbb{R}$ to $\mathbb{R}^2$.} $\alpha:\mathbb{R}\to\mathbb{R}^2$ such that the following conditions hold:
\begin{itemize}
\item[i)] The set $A$ is the image of $\alpha$. 
\item[ii)] The map $\alpha$ is continuous.
\item[iii)] $\lim\limits_{|t|\to\infty} \|\alpha(t)\| = \infty$, where $\|\cdot\|$ is the Euclidean norm.
\item[iv)] $\{t\in\mathbb{R}:\ \exists s\neq t,\ \alpha(s)=\alpha(t)\}$ has measure zero.
\end{itemize}
A relation $\varUpsilon$ is called cursive if the set $\{(p, q) \in \mathbb{R}^2:\ q \in \varUpsilon(p)\}$ is cursive.
\end{definition}
Intuitively speaking, a relation is cursive if it can be drawn on a piece of paper without lifting the pen. The third requirement demands that the drawing will be infinite (in both time directions), and the fourth allows the pen to cross its own path, but forbids it from going over the same line twice. This intuition is the reason we call these relations cursive relations.

Under the assumption that the steady-state I/O relation $k$ of $\Sigma$ is cursive (which is usually the case for dynamical systems of the form \eqref{system_model_single}), we prove the maximality of $\lambda$:

\begin{thm} \label{thm.Maximality}
Let $k$,$\lambda$ be the steady-state I/O relations of the original system $\Sigma$ and the transformed system $\tilde{\Sigma}$ under the transformation $T$, respectively. Suppose $k$ is a cursive relation and $T$ is chosen to monotonize $k$ as in Theorem~\ref{thm.MappingPQIs}. Then,
\begin{enumerate}
\item[i)] $\lambda$ is a maximally monotone relation, and 
\item[ii)] $\tilde{\Sigma}$ is MEIP.
\end{enumerate}
Moreover, if $\lambda$ is a strictly monotone relation, then $\tilde{\Sigma}$ is input-strictly MEIP, and if $\lambda^{-1}$ is a strictly monotone relation, then $\tilde{\Sigma}$ is output-strictly MEIP.
\end{thm}

Before proving the theorem, we prove the following lemma.

\begin{lemma}	\label{lem_cursive_monotone_ralation}
A cursive monotone relation $\varUpsilon$ must be maximally monotone.
\end{lemma}
\begin{proof}
Let $A_\varUpsilon \subseteq \mathbb{R}^2$ be the set associated with $\varUpsilon$, which is cursive by assumption. Let $\alpha$ be the corresponding curve. If $\varUpsilon$ is not maximal, there is a point $(p_0, q_0) \notin A_\varUpsilon$ so that $\varUpsilon \cup \{(p_0, q_0)\}$ is a monotone relation. By monotonicity,
\begin{align*}
A_\varUpsilon \subseteq \{(p, q) \in \mathbb{R},\ (p \ge p_0\text{  and }q \ge q_0) \text{ or }\\ (p\le p_0\text{ and }q\le q_0), (p, q) \neq (p_0, q_0)\}.
\end{align*}
The set on the right hand side has two connected components, namely $\{(p,q):\ p\ge p_0, q\ge q_0, (p,q)\neq(p_0,q_0)\}$ and $\{(p,q):\ p\le p_0, q\le q_0, (p,q)\neq(p_0,q_0)\}$. Since $A_\varUpsilon$ is the image of a continuous map $\alpha$,  it is contained in one of these connected components. Suppose, without loss of generality, it is contained in  $\{(p,q):\ p\ge p_0, q\ge q_0, (p,q)\neq(p_0,q_0)\}$.
It is clear that we can choose the curve $\alpha(t)=(\alpha_1(t),\alpha_2(t))$ so that both functions $\alpha_1,\alpha_2$ are non-decreasing, as $\varUpsilon$ is monotone. Thus, we must have $\alpha_1(0) \ge \lim_{t\to -\infty} \alpha_1(t) \ge p_0 ,\ \alpha_2(0) \ge \lim_{t\to -\infty} \alpha_2(t) \ge q_0$. However, these inequalities imply that $\|\alpha(t)\|=\sqrt{\alpha_1(t)^2+\alpha_2(t)^2}$ remains bounded as $t\to -\infty$. This contradicts the assumption that $\varUpsilon$ was cursive, hence it must be maximally monotone.
\end{proof}

We are now ready to prove Theorem \ref{thm.Maximality}.
\begin{proof}
By definition of MEIP and Lemma~\ref{lem_cursive_monotone_ralation}, it is enough to show that if $k$ is cursive, then so is $\lambda$. Let $\mathcal{A}_{k}$ be the set associated with $k$, and $\mathcal{A}_{\lambda}$ be the set associated with $\lambda$. Note that $(\tilde{\mathrm u},\tilde{\mathrm y})$ is a steady-state of $\tilde{\Sigma}$ if and only if $({\mathrm u},{\mathrm y})$ is a steady-state of $\Sigma$, where the I/O pairs are related by the transformation $T$. Thus, $\mathcal{A}_{\lambda}$ is the image of $\mathcal{A}_{k}$ under the invertible linear map $T$. Since $k$ is cursive, we have an associated curve $\alpha:\mathbb{R}\to\mathbb{R}^2$ plotting $\mathcal{A}_{k}$. We define the curve $\beta(t)=T(\alpha(t))$. We claim that the curve $\beta$ proves that $\mathcal{A}_{\lambda}$, and hence $\lambda$, 
is cursive. Indeed, it is clear that $\mathcal{A}_{\lambda}$ is the image of $\beta$. Furthermore, $\beta$ is continuous as a composition of the continuous maps $T$ and $\alpha$. The third property in Definition~\ref{defn_cursive_relation} holds as 
$\lim_{|t|\to\infty}||\beta(t)|| \ge \lim_{|t|\to\infty}\underline{\sigma}(T)||\alpha(t)||=\infty$, where we note that $T$ is invertible, hence $\underline{\sigma}(T)$, the minimal singular value of $T$, is positive. Lastly, the fourth property in 
Definition~\ref{defn_cursive_relation} holds as $\beta(t)=\beta(s)$ if and only if $\alpha(t)=\alpha(s)$, as $T$ is invertible. Thus, the set $\{t:\ \exists s\neq t, \beta(t)=\beta(s)\}$ is the same as the one for $\alpha$, having measure zero. 

Lastly, we need to show that if $\lambda$ is strictly monotone, then $\tilde{\Sigma}$ is strictly MEIP. A strictly monotone relation $\lambda$ is achieved when taking $\nu^\prime > 0, \rho^\prime \ge 0$ in Proposition \ref{prop.MonInducedPassivation}, so we conclude that $\tilde{\Sigma}$ is EI-IOP($0,\nu^\prime$) for some $\nu^\prime > 0$, and thus input-strictly MEIP as its input-output relation, $\lambda$, is maximally monotone. The case in which $\lambda^{-1}$ is strictly monotone is dealt similarly.
\end{proof}

Before moving to the network optimization framework, we wonder how common are cursive relations. Obviously, all stable linear systems have cursive steady-state I/O relations, as their steady-state I/O relations form a line inside $\R^2$. As a more general example, we prove the following proposition for a class of input-affine nonlinear systems:
\begin{proposition}
Consider the system $\Upsilon$ governed by the ODE $\dot{x} = -f(x) + g(x)u,\ y=h(x)$ for some $\mathcal{C}^{1}$ smooth functions $f,g$ and a continuous function $h$ such that $g>0$. Assume that either $f/g$ or $h$ is strictly monotone ascending, and that either $\lim_{s\to\pm\infty} |h(s)| = \infty$ or $\lim_{s\to\pm\infty} |f(s)/g(s)| = \infty$. Then the system $\Upsilon$ has a cursive steady-state I/O relation.
\end{proposition}
\begin{proof}
In steady-state, we have $\dot{x} = 0$, thus we have $f(\mathrm x) = g(\mathrm x)\mathrm u$. Moreover, $\mathrm y = h(\mathrm x)$ in steady-state. Thus the steady-state input-output relation can be parameterized as $(f(\sigma)/g(\sigma),h(\sigma))$ for the parameter $\sigma \in \R$. Consider the curve $\alpha:\R \to \R^2$ defined by $\alpha(\sigma) = (f(\sigma)/g(\sigma),h(\sigma))$. Then the steady-state relation is the image of $\alpha$, which is continuous. The norm of $\alpha$ is equal to $\sqrt{(f(\sigma)/g(\sigma))^2+h(\sigma)^2}$, so the assumption on the limit shows that $\lim_{|t|\to\infty} ||\alpha(t)|| = \infty$. Lastly, by strict monotonicity, the curve $\alpha$ is one-to-one. Thus the steady-state input-output relation is cursive.
\end{proof}
\begin{remark}
The strict monotonicity assumption can easily be relaxed$-$it shows that the curve $\alpha(t) = (f(t)/g(t),h(t))$ is one-to-one, but in practice we may have a non-self-intersecting curve, which can behave very wildly in each coordinate. 
Moreover, non-self-intersecting is a stronger requirement then needed, we only 
need that the ``self-intersecting set" is of measure zero.
\end{remark}

As we showed that cursive relations appear for a wide class of systems, we conclude the network optimization framework for EIPS) agents by Theorem \ref{thm_network_flow_problems} and Theorem \ref{thm.MappingPQIs}.
\begin{thm}\label{thm_transformed_network_optimization_problems}
Consider the diffusively-coupled network $(\pmb{\Sigma}, \pmb{\Pi}, \mathcal{G})$, and suppose the agents ${\Sigma}_i$ are EI-IOP($\rho_i,\nu_i$) with cursive steady-state I/O relations $k_i$, and that the controllers are MEIP with integral functions $\Gamma_e$. Let $\mathcal{J}= \mathrm{diag}(T_1,T_2,\ldots,T_{|\mathcal{V}|})$ be a linear transformation, where $T_i$ is chosen as in Theorem \ref{thm.MappingPQIs} so that $k_i^{-1}$ is transformed into a strictly monotone relation by applying $T_i$. Then the transformed network $(\pmb{\tilde{\Sigma}}, \pmb{\Pi},\mathcal{G})$ converges, and the steady-state limits $(\pmb{\tilde{\rm u}}, \pmb{\tilde{\rm y}},\pmb{\upzeta},\pmb{\upmu})$ are minimizers of the following dual network optimization problems:
\begin{center}
\begin{tabular}{c||c}
{\bf TOPP} & {\bf TOFP}\\
\hline
\parbox{1cm}{\begin{subequations}
\begin{alignat}{2}
\nonumber & \!\min_{\pmb{\tilde{\rm y}}, \pmb{\upzeta}} &\qquad& \pmb{\Lambda}^\star(\pmb{\tilde{\rm y}}) + \pmb{\Gamma}(\pmb{\upzeta})\\
\nonumber &{s.t.} &      & \E^\top{\pmb{\tilde{\rm y}}} = \pmb{\upzeta}
\end{alignat}
\end{subequations}} &  \parbox{1cm}{\begin{subequations}
\begin{alignat}{2}
\nonumber &\!\min_{\pmb{\tilde{\rm u}}, \pmb{\upmu}} &\qquad& \pmb{\Lambda}(\pmb{\tilde{\rm u}}) + \pmb{\Gamma}^\star(\pmb{\upmu})\\
\nonumber &{s.t.} &      & \pmb{\tilde{\rm u}} = - \E\pmb{\upmu}
\end{alignat}
\end{subequations}}
\end{tabular}
\end{center}
where $\pmb{\Gamma}(\pmb{\upzeta}) = \sum_{e\in\EE} \Gamma_e(\upzeta_e)$, $\pmb{\Lambda}(\pmb{\mathrm u}) = \sum_{i\in\V} \Lambda_i(\mathrm u_i)$, and $\pmb{\Lambda}_i$ is the integral function associated with the maximally monotone relation ${\lambda}_i$, obtained by applying $T_i$ on $k_i$.
\end{thm}

For the special cases in which the original EI-IOP($\rho,\nu$) agents have integral 
functions, we can use the discussion succeeding Proposition~\ref{prop_connection_between_relations}, connecting the original and the transformed integral functions, to prescribe (TOPP) and (TOFP) in terms 
of (OPP) and (OFP). It is worth noting that (TOPP) and (TOFP) can be viewed as regularized 
versions of (OPP) and (OFP), where quadratic terms are added both the the 
agents' integral functions and their duals. This is a generalization of 
\cite{Jain2018} which prescribed the quadratic correction of (OPP) when the 
agents are EI-OP$(\rho)$. The main difference in 
our approach from the one in \cite{Jain2018} is that there, the network 
optimization framework can always be defined, and convexifying it leads to the passivizing transformation. In the case presented here, the simultaneous input- 
and output-shortage of passivity can cause the network optimization framework 
to be undefined, forbidding us from trying to convexify it. Instead, we resort 
to monotonizing the steady-state relation, which in turn induces a passivizing 
transformation. This approach can be seen pictorially in Figure 
\ref{MonotonizationIdea}. In particular, we conclude by re-stating the main result of 
\cite{Jain2018} and providing a proof using the methods introduced here.

\begin{figure}[t!]
	\centering
	\includegraphics[width=6cm]{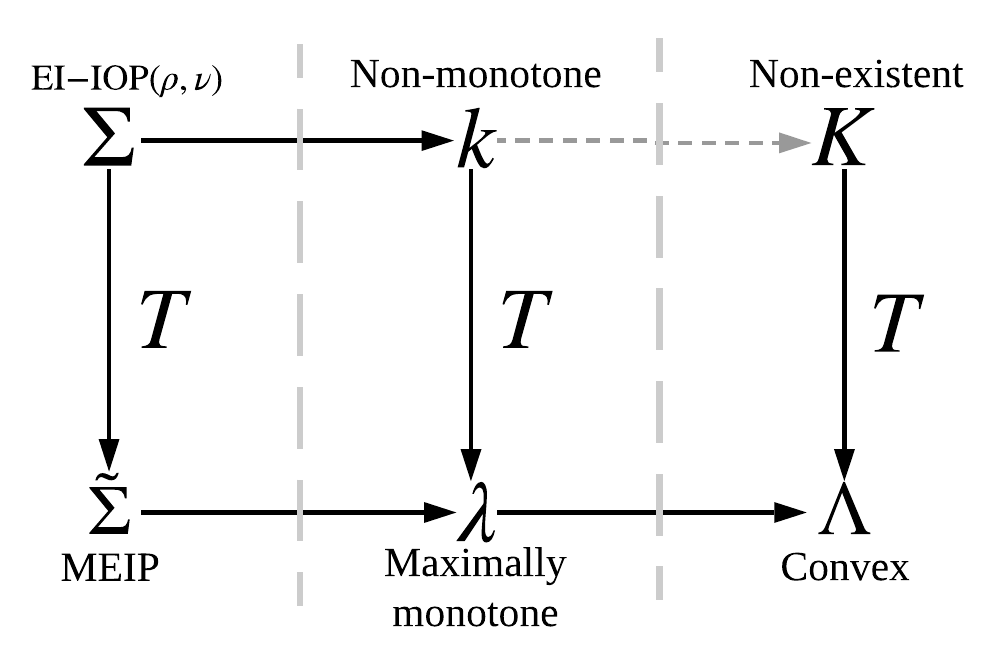}
	\vspace{-3pt}
	\caption{Monotonization, passivation and convexification by the 
	transformation 
		$T$. For general output-passive short systems, convexification is 
		equivalent to 
		passivation. For EI-IOP($\rho,\nu$) systems, integral functions do not  
		necessarily exist, so monotonization of the steady-state relation is 
		equivalent 
		to passivation.} 
	\label{MonotonizationIdea}\vspace{-15pt}
\end{figure}

\begin{corollary}
Let  $(\pmb{\Sigma}, \pmb{\Pi}, \mathcal{G})$ be a diffusively-coupled network, and suppose the agents have cursive steady-state I/O relations $k_i$, and that the controllers are MEIP with integral function $\Gamma_e$. Let $\mathcal{J}= \mathrm{diag}(T_1,T_2,\ldots,T_{|\mathcal{V}|})$ be as in Theorem \ref{thm_transformed_network_optimization_problems}.
\begin{enumerate}
\item[i)] If the agents $\Sigma_i$ are EI-OP($\rho_i$), and the relations $k_i^{-1}$ have integral functions $K_i^\star$, then we can take $T_i = \left[\begin{smallmatrix}1 & \beta_i \\ 0 & 1\end{smallmatrix}\right]$ for $\beta_i > - \rho_i$, and the cost function of (TOPP) is $\pmb{K}^\star(\pmb{\mathrm y}) + \pmb{\Gamma}(\pmb{\upzeta}) + \frac{1}{2}\pmb{\rm y}^\top\mathrm{diag}(\pmb{\beta})\pmb{\rm y}$, where $\pmb{K}^\star(\pmb{\rm y}) = \sum_{i\in\V} K_i^\star(\mathrm y_i)$.
\item[ii)] If the agents $\Sigma_i$ are EI-IP($\nu_i$), and the relations $k_i$ have integral functions $K_i$, then we can take $T_i = \left[\begin{smallmatrix}1 & 0 \\ \beta_i & 1\end{smallmatrix}\right]$ for any $\beta_i > - \nu_i$, and the cost function of (TOFP) is $\pmb{K}(\pmb{\mathrm u}) + \pmb{\Gamma}^\star(\pmb{\upmu}) + \frac{1}{2}\pmb{\rm u}^\top\mathrm{diag}(\pmb{\beta})\pmb{\rm u}$, where $\pmb{K}(\pmb{\rm y}) = \sum_{i\in\V} K_i(\mathrm u_i)$.
\end{enumerate}
\end{corollary}
\begin{proof}
We only prove the first case, as the proof second case is completely analogous. Each agent is EI-OP($\rho_i$), so that the associated PQI is $0\le \xi\chi - \rho_i \chi^2$. We take any $\beta_i>-\rho_i$ and look for $T_i$ transforming this PQI into $0 \le \xi\chi - (\rho_i+\beta_i)\chi^2$, which implies output-strict MEIP. We build $T_i$ according to Theorem \ref{thm.MappingPQIs}, taking $(\xi_1,\chi_1) = (1,0), (\xi_2,\chi_2) = (\rho_i,1), (\xi_3,\chi_3) = (1,0)$ and $(\xi_4,\chi_4) = (\rho_i+\beta_i,1)$.  We note that $(\xi_1+\chi_1,\xi_2+\chi_2) = (1+\rho_i,1)$ satisfies 
$
\chi\xi - \rho_i\chi^2 = 1+\rho_i - \rho_i = 1 \ge 0,
$
meaning that $(\xi_1+\chi_1,\xi_2+\chi_2)$ satisfies the first PQI. Similarly, $(\xi_3+\chi_3,\xi_4+\chi_4)$ satisfies the second PQI. We thus take:
\begin{align*}
 T_i &= \left[\begin{matrix}\xi_3 & \xi_4 \\ \chi_3 & \chi_4\end{matrix}\right]\left[\begin{matrix}\xi_1 & \xi_2 \\ \chi_1 & \chi_2 \end{matrix}\right]^{-1}\\ &= \left[\begin{matrix}1 & \rho_i+\beta_i \\ 0 & 1\end{matrix}\right]\left[\begin{matrix}1 & \rho_i \\ 0 & 1 \end{matrix}\right]^{-1} = \left[\begin{matrix}1 & \beta_i \\ 0 & 1\end{matrix}\right],
\end{align*}
which proves the first part. As for the second part, Table \ref{table} implies 
that the steady-state relation $\lambda_i$ of the transformed system is given by 
$\lambda_i^{-1}(\mathrm y_i) = k_i^{-1}(\mathrm y_i) + \beta_i \mathrm y_i$. 
Integrating this equation with respect to $\mathrm y_i$ gives that 
$\Lambda_i^\star (\mathrm y_i) = K_i^\star(\mathrm y_i) + \frac{1}{2}\beta_i 
\mathrm y_i^2$. Using $\pmb{K}^\star(\pmb{\rm y}) = \sum_{i\in\V} 
K_i^\star(\mathrm y_i)$ and $\pmb{\Lambda}^\star(\pmb{\rm y}) = \sum_{i\in\V} 
\Lambda_i^\star(\mathrm y_i)$ gives that $\pmb{\Lambda}^\star(\pmb{\rm y}) = 
\pmb{K}^\star(\pmb{\rm y}) + \frac{1}{2}\pmb{\rm 
y}^\top\mathrm{diag}(\pmb{\beta})\pmb{\rm y}$, completing the proof.
\end{proof}

\section{Case Studies}\label{An Example}
This section presents two examples illustrating the theoretical results proposed in this paper. The first example deals with a collection of EIPS linear and time-invariant systems, and exemplifies the application of Algorithm \ref{alg.passivation} on a specific system. The second example describes a network of gradient systems with non-convex potential functions, exemplifying the results of Section \ref{Cursive}.

\subsection{Linear and Time Invariant Systems} \label{exam.LTI_RelDeg}
Consider a linear time-invariant system $\Sigma$ with a transfer function of the form $G(s) = \frac{\varsigma}{s^2+as+b}$, where $a,b,\varsigma\in \mathbb{R}$ and $\varsigma \neq 0$. We consider the case in which $a>0$, where $a$ is equal to minus the sum of the poles of the system. This case occurs when both poles are stable, or only one pole is stable. Examples of such systems include the oscillations of a ship at sea \cite{Fortuna1996}, robot elbow actuators \cite[p. 487]{Dorf2011}, and suspended mobile remote cameras, as used in sports events \cite[p. 881]{Dorf2011}. The prior of the three has two stable poles, where the latter two only have one stable pole. If both poles are stable, then the system has a finite $\mathcal{L}_2$-gain and can be stabilized using the small-gain theorem \cite{Khalil2002}. Otherwise, the system does not have a finite $\mathcal{L}_2$-gain.

According to Theorem \ref{eq.thm8}, in this case, $p(s) = \varsigma$ and $q(s) = s^2 + as + b$, so that $\deg p = 0 < \deg q = 2$, and the degree of $q(s)+\lambda p(s)$ is two. If we choose $\lambda = \frac{0.25a^2-b}{\varsigma}$, then $q(s)+\lambda p(s) = s^2 + as + 0.25a^2 = (s+0.5a)^2$, which has a double stable pole at $s=-0.5a$. Moreover, computing $\mu = \sup_{\omega\in\mathbb{R}} \left|\frac{\varsigma}{(j\omega+0.5a)^2}\right|^2 + \frac{1}{4}$ gives $\mu = \frac{4\varsigma}{a^2} + \frac{1}{4}$. Thus, the system $\Sigma$ is EI-IOP($\rho,\nu$) for $\rho =- \frac{\lambda(1+\lambda\mu)}{1+2\lambda\mu}$ and $\nu = -\frac{\mu}{1+2\lambda\mu}$.

As a specific example, consider the linear and time-invariant system $\Sigma$ with the transfer function $G(s) = \frac{0.75}{s^2+2s-2}$, which has a stable pole at $s=-1-\sqrt{3} \approx -2.73$ and an unstable pole at $s=\sqrt{3} - 1 \approx 0.73$. We note this system is not finite $\mathcal{L}_2$-gain, nor input-passive short, as it has an unstable pole, nor output-passive short, as it has a relative degree of $2$ \cite{Khalil2002}. For this system, we have $\lambda = 4$ and $\mu = 1$, which in turn give $\rho = -\frac{20}{9}$ and $\nu = -\frac{1}{9}$.

We now passivize $\Sigma$ by applying Algorithm \ref{alg.passivation}. We first note that $(\xi_1,\chi_1) = (5,-1)$ and $(\xi_2,\chi_2) = (-4,1)$ are two non-colinear solutions of $-\nu\xi^2+\xi\chi-\rho\chi^2 = \frac{1}{9}(4\chi+\xi)(5\chi+\xi) = 0$. Choosing $\rho^\prime = \nu^\prime = 0$, and the corresponding non-colinear solutions $(\xi_3,\chi_3) = (1,0)$ and $(\xi_4,\chi_4)=(0,1)$ to the equation $-\rho^\prime \xi^2 + \xi\chi -\nu^\prime \chi^2 = 0$, we compute:

\begin{align*}
\alpha_1 &= -\rho(\xi_1+\xi_2)^2+(\xi_1+\xi_2)(\chi_1+\chi_2)-\nu(\chi_1+\chi_2)^2\\& = \frac{1}{9} > 0 \\
\alpha_2 &= -\rho^\prime(\xi_3+\xi_4)^2+(\xi_3+\xi_4)(\chi_3+\chi_4)-\nu^\prime(\chi_3+\chi_4)^2 \\&= 1 > 0.
\end{align*}
Thus, the transformation $T_1$, as defined in \eqref{eq.T1T2}, passivizes the system $\Sigma$. A simple computation shows that $T_1 = \left[\begin{smallmatrix} 1 & 4 \\ 1 & 5\end{smallmatrix}\right]$, implying that the transformed input and output are given by $\tilde u = u+4y, \tilde y = u+5y$. If we let $U(s),Y(s),\tilde U(s),\tilde Y(s)$ be the Laplace transforms of $u,y,\tilde u,\tilde y$ respectively, then the connections $\tilde U(s) = U(s)+4Y(s) = (1+4G(s))U(s)$ and $\tilde Y(s) = U(s) + 5Y(s) = (1+5G(s))U(s)$ show that the transfer function of the transformed system $\tilde{\Sigma}$ is equal to
\begin{align*}
\tilde G(s) = \frac{\tilde Y(s)}{\tilde U(s)} = \frac{s^2+2s+3}{s^2+2s+2}.
\end{align*}
This transfer function, and therefore $\tilde{\Sigma}$, is passive, and is in fact input-strictly passive with index $0.9$ and output-strictly passive with parameter $\frac{2}{3}$, as can be verified by the MATLAB command ``{\rm getPassiveIndex}." The fact that $\tilde{\Sigma}$ is \emph{strictly} passive follows from our choice of $\lambda$, which requires all zeros of a certain polynomial to be in the open left-half plane, not allowing any to be on the imaginary axis.

\subsection{A Network of Gradient Systems with Non-Convex Potentials}
\begin{figure*}[t!]
	\centering
	\subfigure[$k_i$]{\includegraphics[width=3.6cm]{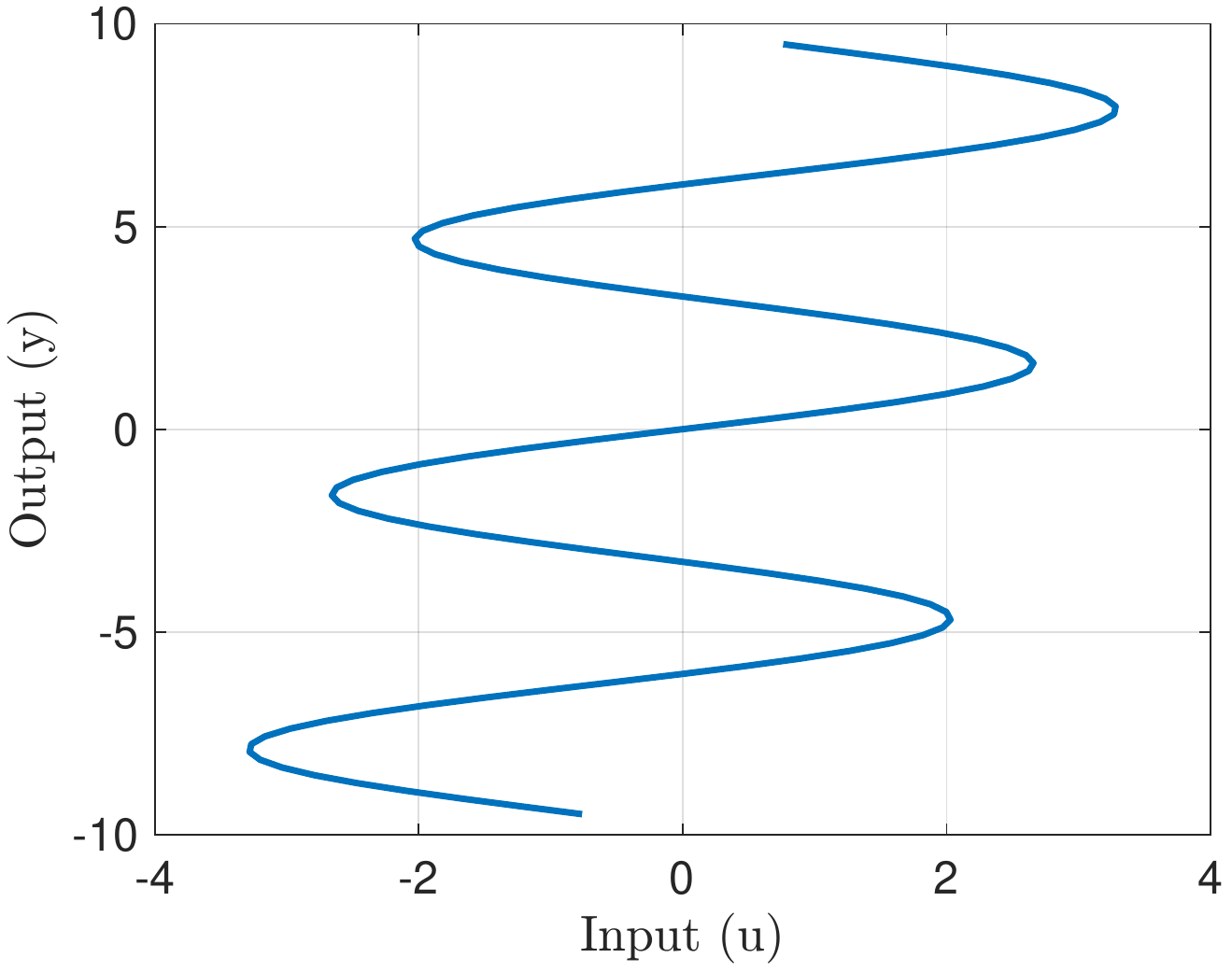}}\hspace{0.5cm}
	\subfigure[$k_i^{-1}$]{\includegraphics[width=3.6cm]{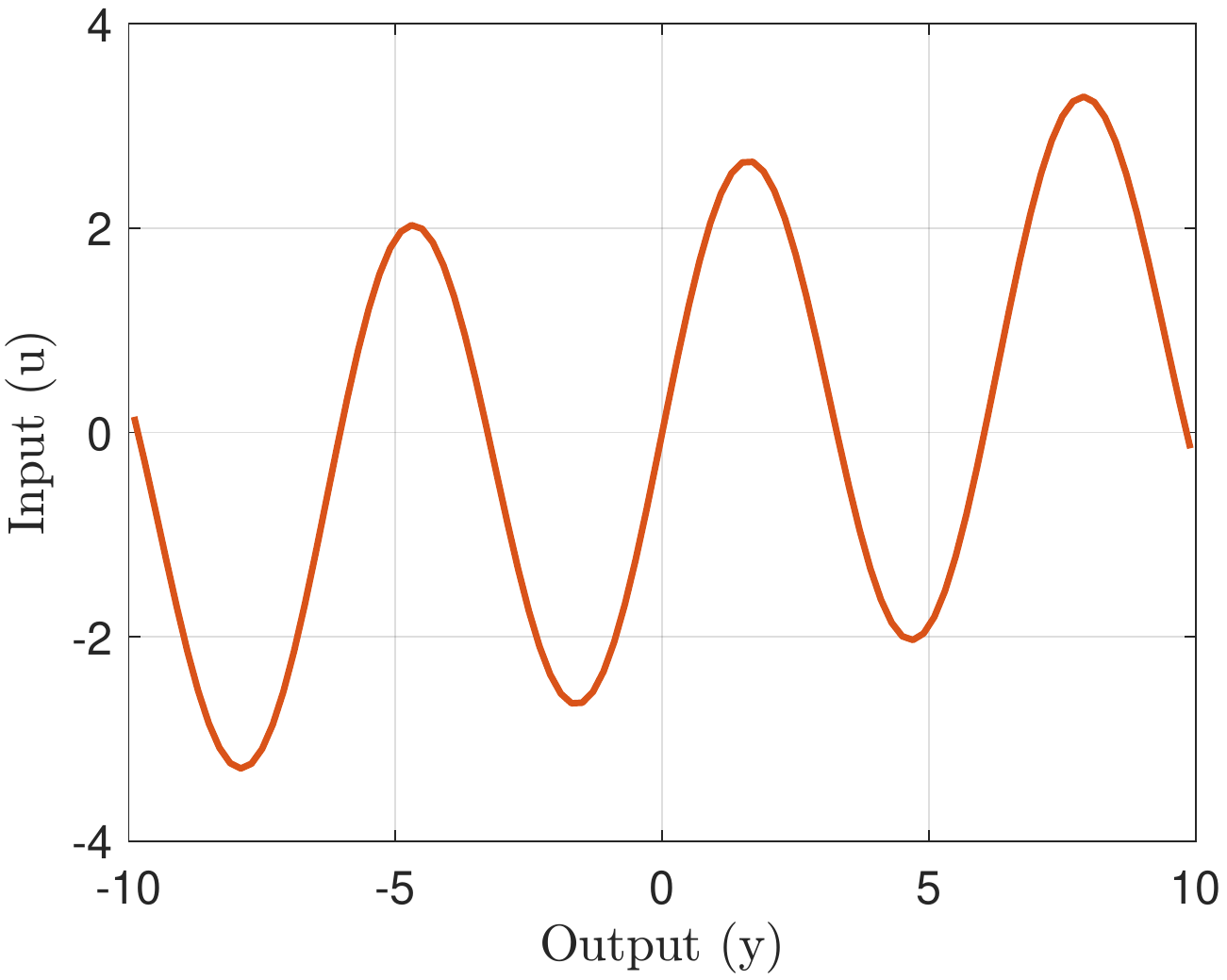}}\hspace{0.5cm}
	\subfigure[$K^\star_i$]{\includegraphics[width=3.6cm]{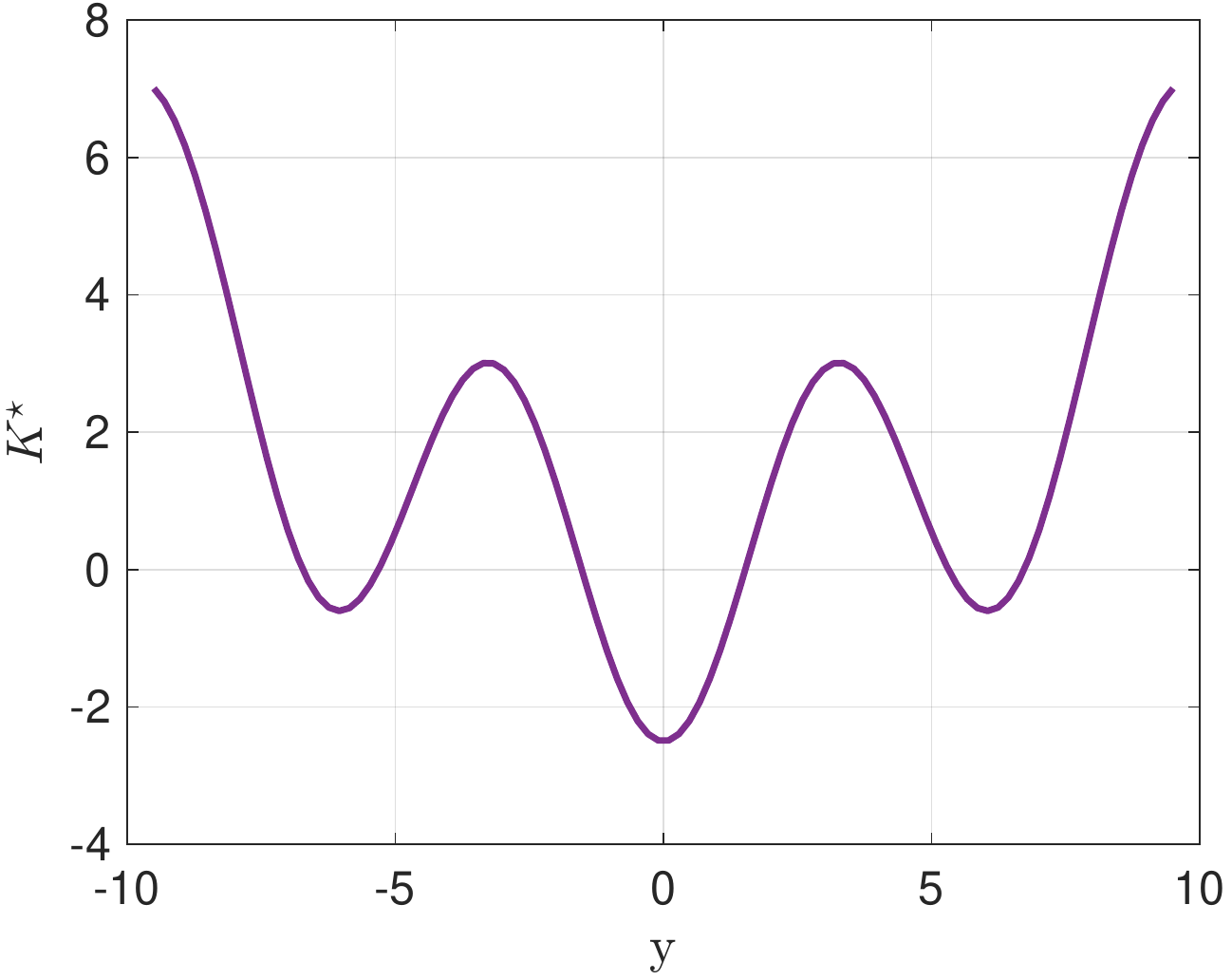}}
	\caption{Steady-state relations and the associated integral function of the EIPS system $\Sigma_i$. Both $k_i$ and $k_i^{-1}$ are cursive but non-monotone and the dual integral function $K^\star_i$ is non-convex.} 
	\label{fig_original_systems} \vspace{-5pt}
\end{figure*}
We consider a class of networked nonlinear gradient systems, described by 
\begin{equation}\label{eq_gradient_system}
\Sigma_i:~\dot{x}_i = -\frac{\partial U(x_i)}{\partial x_i} + u_i;~~y_i = x_i,~~~ i = 1, \ldots, |\mathcal{V}|,
\end{equation} 
where the inputs $u_i$ are given by
\begin{equation}\label{eq_gradient_system_control}
u_i = G\sum_{j \in \mathcal{N}_i} (x_j - x_i),~~~i = 1, \ldots, |\mathcal{V}|,
\end{equation}
where $G > 0$ is the controller gain, $\mathcal{N}_i$ denotes the neighbors of agent $i$, and $U$ is a scalar potential function with $U(\sigma) > 0, \sigma \neq 0, U(0) = 0$. Such classes of systems are important because of their applications in both biological and multi-agent systems, and are inspired from \cite{Scardovi2009}. As discussed in \cite{Scardovi2009}, \eqref{eq_gradient_system} loosely describes the dynamics of a group of bacteria performing chemotaxis (where $x_i$ is the position of the bacteria) in response to chemical stimulus, such as the concentration of chemicals in their environment, to find food (for example, glucose) by swimming towards the highest concentration of food molecules. Other possible applications include vehicle networks that must efficiently climb gradients to search for a source by measuring its signal strength in a spatially distributed environment. Note that this is a diffusively-coupled systems, with agents $\Sigma_i$ and static gains $G$ as edge controllers. It's easy to verify that the static controllers $\Pi_e$ are MEIP and that their I/O relation $\gamma_e$ is a straight line passing through origin in the $({\upzeta}_e, {\upmu}_e)$ plane. 

Let the potential $U$ be given by $U(x_i) = r_1(1 - \cos{x_i}) + \frac{1}{2}r_2x^2_i, r_1 > 0, r_2 > 0$. Thus $\frac{\partial U}{\partial x_i} = r_1\sin{x_i} + r_2x_i$ and the Hessian is $\frac{\partial^2 U}{\partial x^2_i} = r_1\cos{x_i} + r_2 \geq (r_2-r_1)$. Note that the steady-state I/O relation $k_i$ of $\Sigma_i$ is given by the planar curve ${\rm u}_i = r_1\sin\sigma + r_2\sigma$; ${\rm y}_i = \sigma$, parameterized by the variable $\sigma$.  

\begin{figure}[t!]
	\centering
	\subfigure[$\lambda_i$]{\includegraphics[width=2.8cm]{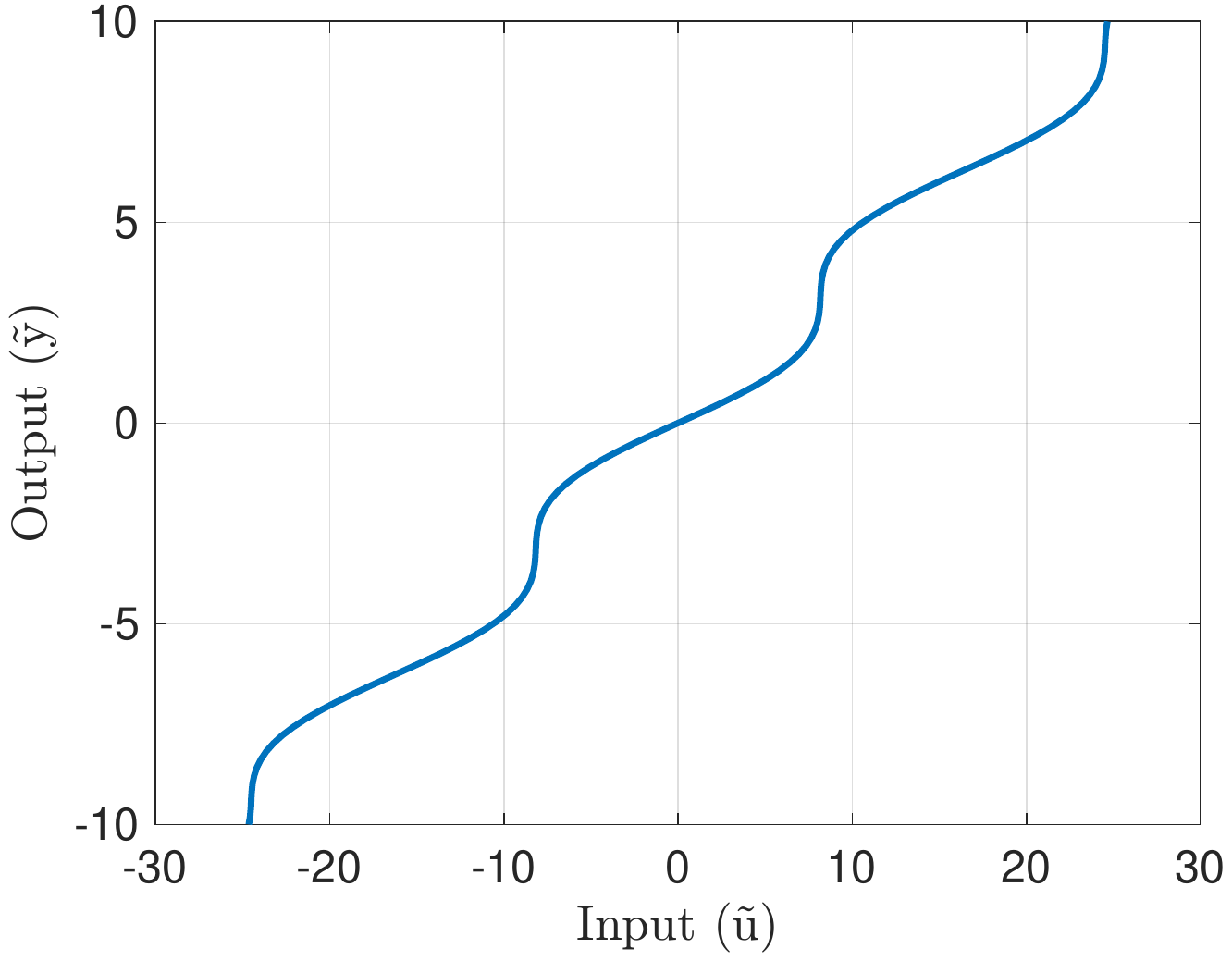}}\hspace{0.2cm}
	\subfigure[$\lambda_i^{-1}$]{\includegraphics[width=2.8cm]{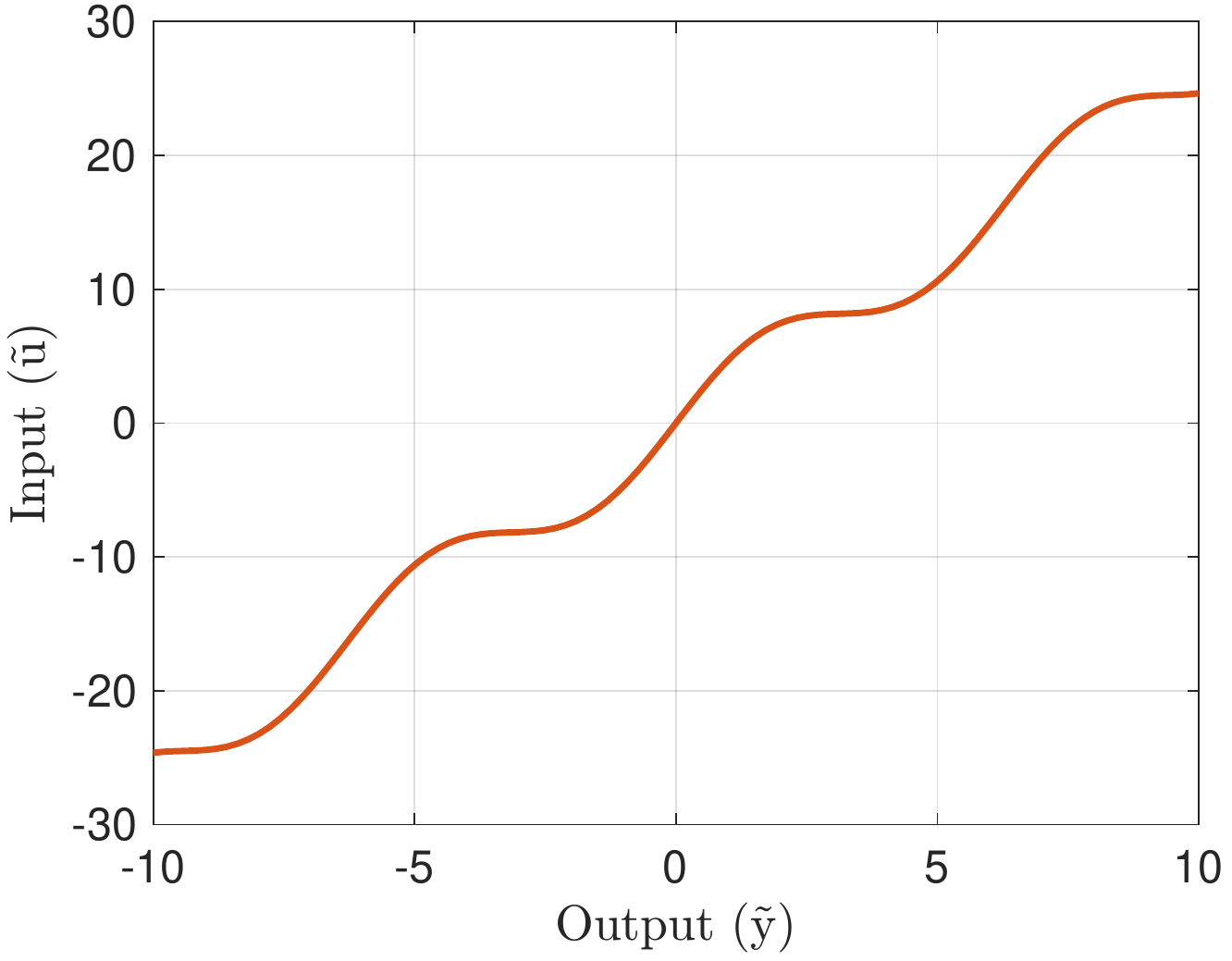}}	
	\caption{Steady-state I/O  relations of the transformed system $\tilde{\Sigma}_i$. Both the relations are maximally monotone.} 
	\label{fig_transformed_system_maps} \vspace{-5pt}
\end{figure}

\begin{figure}[t!]
	\centering
	\subfigure[$\Lambda_i$]{\includegraphics[width=2.8cm]{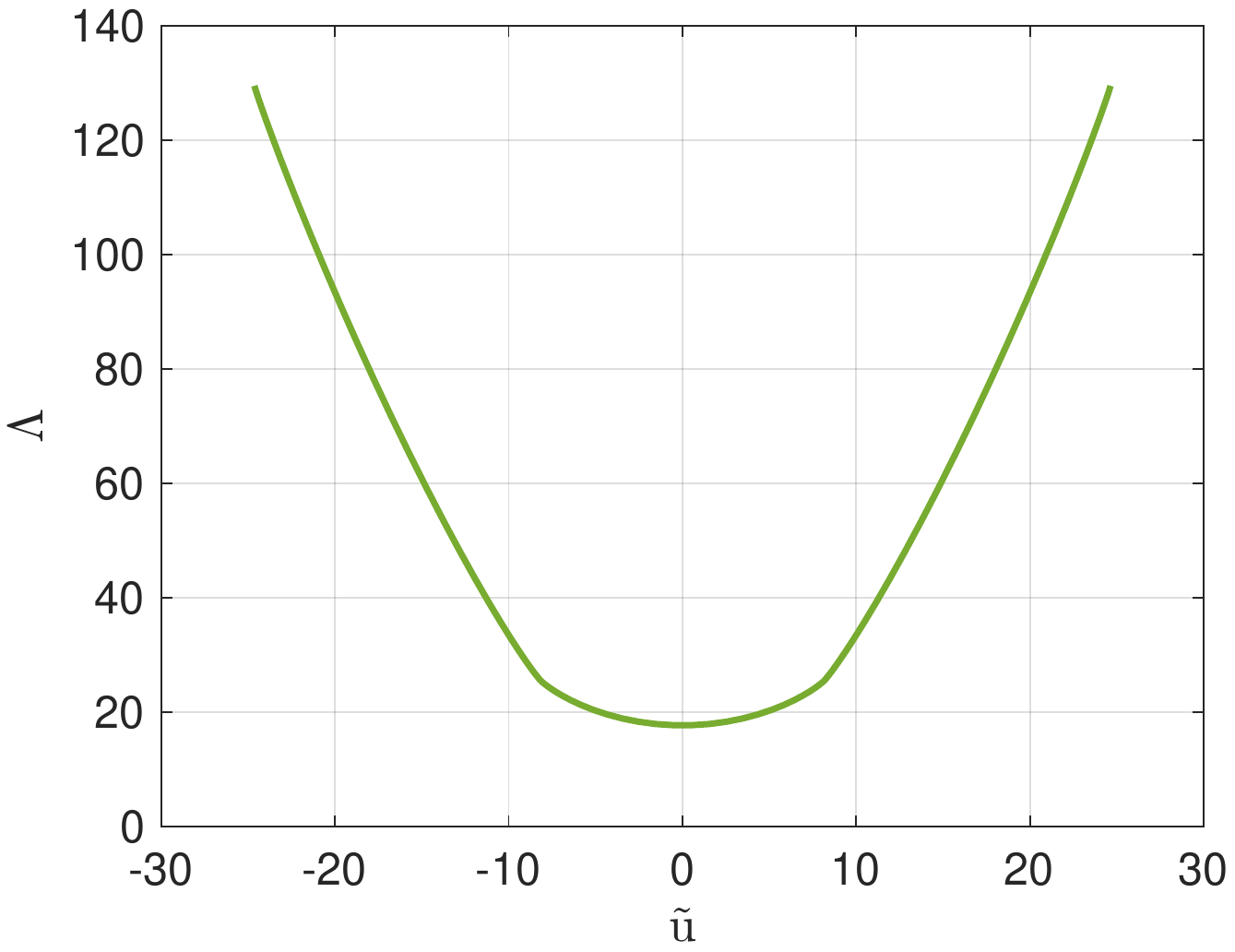}}\hspace{0.2cm}
	\subfigure[$\Lambda_i^\star$]{\includegraphics[width=2.8cm]{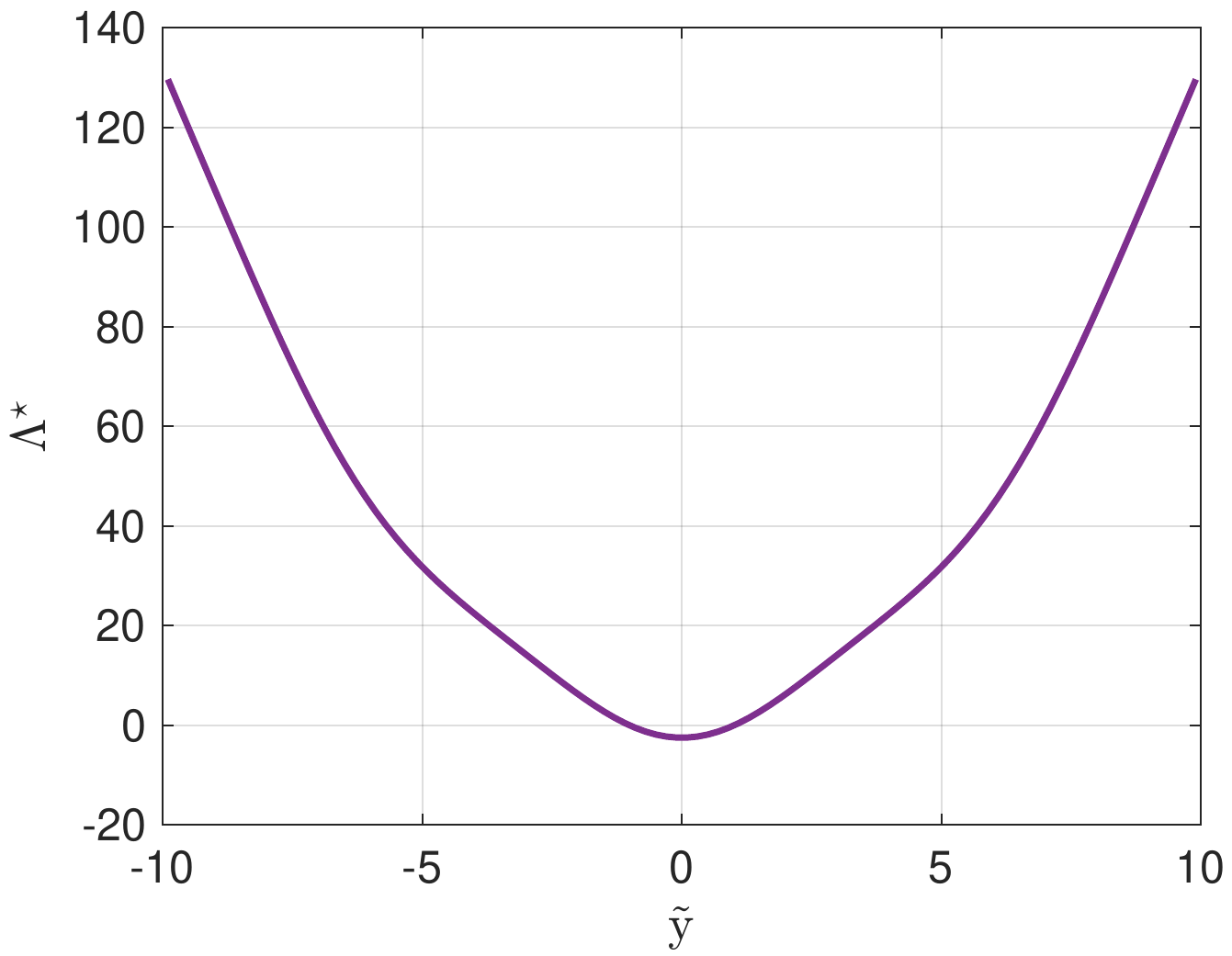}}
	\caption{Integral functions associated to steady-state I/O relations of the transformed system $\tilde{\Sigma}_i$. Both $\Lambda_i$ and $\Lambda_i^\star$ are strictly convex and attains their minimum at the steady-states of the network.} 
	\label{fig_transformed_system_integrals} \vspace{-8pt}
\end{figure}

\begin{figure}[t!]
	\centering
	\subfigure[$\Sigma$]{\includegraphics[width=2.8cm]{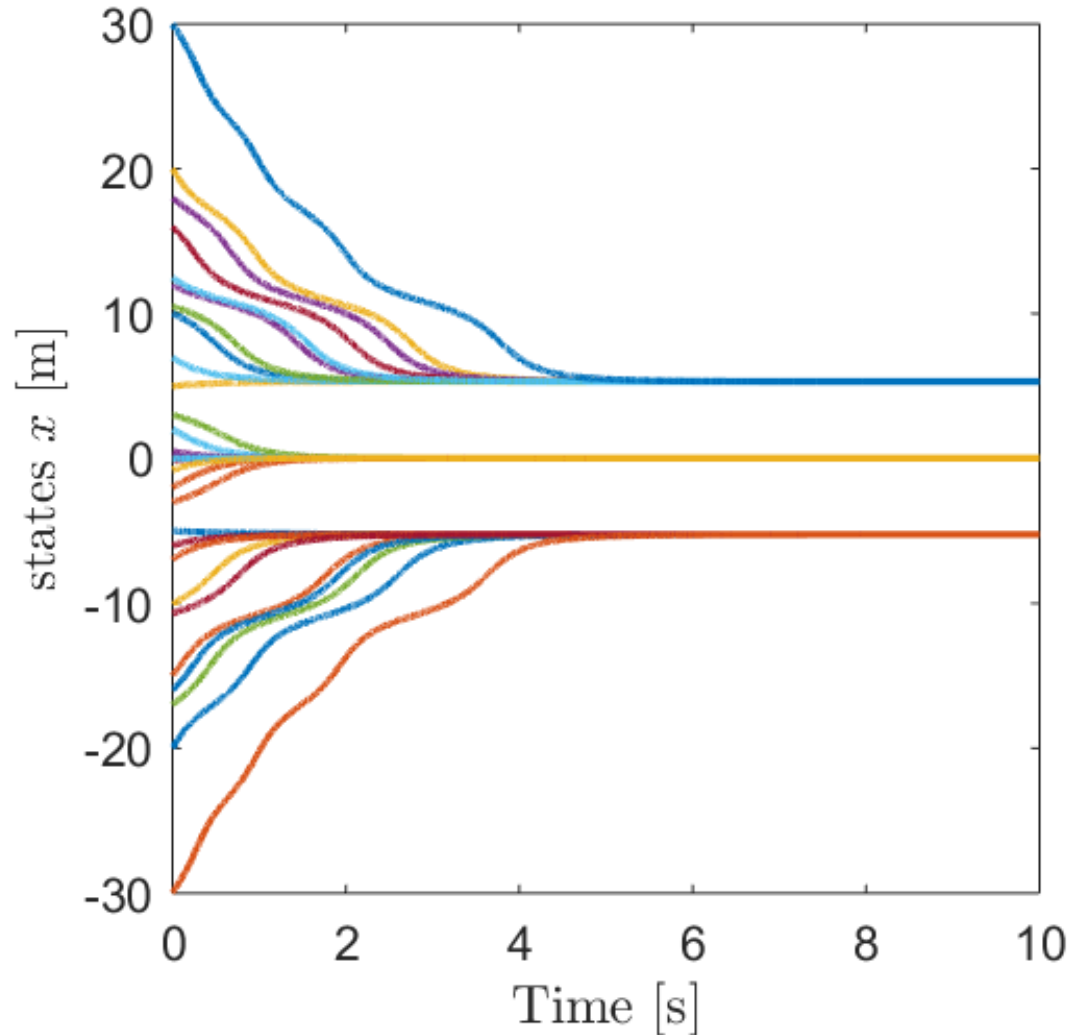}}\hspace{0.2cm}
	\subfigure[$\tilde{\Sigma}$]{\includegraphics[width=2.8cm]{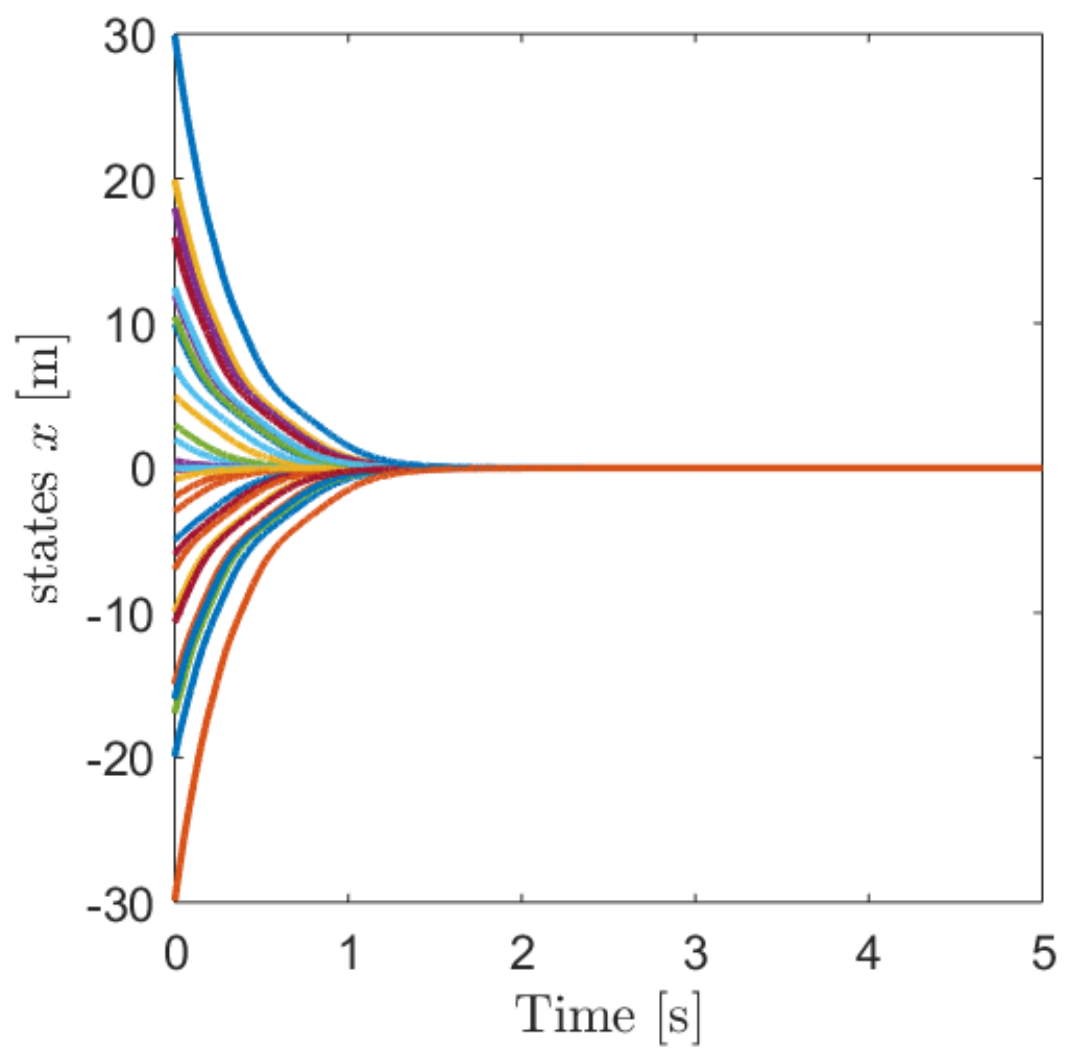}} 
	\caption{States of the systems $\Sigma$ and $\tilde{\Sigma}$ in the diffusively-coupled network interconnection in Figure~\ref{final_network_model}.} \vspace{-5pt}
	\label{fig_States of the systems}
\end{figure}

We choose $r_1 = 2.5, r_2 = 0.1$ and note that $\frac{\partial^2 U}{\partial x^2} \geq \rho {\rm Id}$, with $\rho = (r_2 - r_1) = -2.4$. Thus, the systems $\Sigma_i$ are EI-OP($\rho$) for $\rho = -2.4$, as mentioned in Proposition~\ref{prop_gradient_system}. The steady-state I/O relation $k_i$ is cursive but non-monotone as shown in Figure~\ref{fig_original_systems}(a) and the associated integral function $K_i$ does not exist. The inverse relation $k_i^{-1}$ is also non-monotone as shown in Figure~\ref{fig_original_systems}(b), and the associated integral function $K_i^\star({\rm y}_i) = \frac{1}{2}r_2{\rm y}^2_i - r_1\cos{\rm y}_i$ is non-convex as shown in Figure~\ref{fig_original_systems}(c).  

By exploiting above methodology, we passivize network by choosing an I/O transformation $\mathcal{J}$, such that the conditions in Theorem~\ref{thm_transformed_network_optimization_problems} are satisfied. One of such transformations is given by $\mathcal{J} = T \otimes I_{|\V|}$ with $T = 
\left[\begin{smallmatrix}
    1 & 2.5 \\
    0 & 1 
\end{smallmatrix}\right]
$, which can be found using Theorem \ref{thm.MappingPQIs} ($\otimes$ represents the Kronecker product). The transformed network $(\mathcal{G},\pmb{\tilde{\Sigma}}, \pmb{\Pi})$, having input $\pmb{\tilde{u}} = \pmb{u} + 2.5\pmb{y}$ and output $\pmb{\tilde{y}} = \pmb{y}$, has agents that are equilibrium-independent output-strictly passive with passivity index $\tilde{\rho} = 0.1 > 0$ (Theorem~\ref{thm.MappingPQIs}). The steady-state I/O relation $\lambda_i$ of each transformed agent $\tilde{\Sigma}_i$ is given by a planar curve $\tilde{\rm u}_i = r_1\sin\sigma + (r_1 + r_2)\sigma$; $\tilde{\rm y}_i = \sigma$, parameterized by the variable $\sigma$, which is maximally monotone as shown in Figure~\ref{fig_transformed_system_maps}(a), and the associated integral function $\Lambda_i$ is strictly convex as in Figure~\ref{fig_transformed_system_integrals}(a), which we plotted using MATLAB function ``cumtrapz". The inverse relation $\lambda^{-1}_i$ is also maximally monotone as shown in Figure~\ref{fig_transformed_system_maps}(b), and the associated integral function $\Lambda_i^\star = \frac{1}{2}(r_1 + r_2)\tilde{\rm y}^2_i - r_1\cos\tilde{\rm y}_i$  is strictly convex as shown in Figure~\ref{fig_transformed_system_integrals}(b).  

The outputs $\pmb{y}$ of the systems are plotted in Figure~\ref{fig_States of the systems} for the above both cases. For the original systems $\pmb{\Sigma}$, there exists a clustering phenomenon as shown in Figure~\ref{fig_States of the systems}(a), which does not corresponds to the minima of the integral function $K_i^\star$ in Figure~\ref{fig_original_systems}(c). However, for the transformed systems $\pmb{\tilde{\Sigma}}$, one can observe from Figure~\ref{fig_transformed_system_integrals} that the minimum of integral functions $\Lambda_i$ and $\Lambda_i^\star$ occurs at the steady-state of the transformed system $\pmb{\tilde{\Sigma}}$, that is, $\pmb{\tilde{\rm u}} = 0$, $\pmb{\tilde{\rm y}} = 0$, as expected.

\section{Conclusions}\label{conclusions}
In this paper, we considered networks of equilibrium-independent $(\rho,\nu)$-passive systems, and constructed a network optimization framework for their analysis. The first step was considering their steady-state I/O relations, which are not necessarily monotone, and monotonizing them using a linear transformation. This was done by a geometric understanding of the quadratic inequalities satisfied by said steady-state I/O relations. We later showed that this transformation actually passivizes the 
agents with respect to any equilibrium, culminating in Algorithm \ref{alg.passivation} for passivation of equilibrium-independent $(\rho,\nu)$-passive systems. We also studied the implementation of these transformations, connecting the original steady-state I/O relation to the transformed one. The last barrier from proving that the transformed agents are MEIP was maximality of the monotonized steady-state relation, which was tackled using the notion of cursive relations. We compared the suggested methods to similar works, and presented case studies demonstrating the constructed framework. Future research might extend this framework to MIMO agents, and will need to extend the 
geometric understanding of the quadratic inequalities, as well as the notion of 
cursive relations, to systems of higher dimensions.

\section*{Acknowledgments}
The authors would like to gratefully acknowledge Prof. Panos Antsaklis for his helpful discussions, comments, and suggestions on this work. 
\bibliographystyle{IEEEtran} 
\bibliography{References_New}

\begin{thebibliography}{10}
\providecommand{\url}[1]{#1}
\csname url@samestyle\endcsname
\providecommand{\newblock}{\relax}
\providecommand{\bibinfo}[2]{#2}
\providecommand{\BIBentrySTDinterwordspacing}{\spaceskip=0pt\relax}
\providecommand{\BIBentryALTinterwordstretchfactor}{4}
\providecommand{\BIBentryALTinterwordspacing}{\spaceskip=\fontdimen2\font plus
\BIBentryALTinterwordstretchfactor\fontdimen3\font minus
  \fontdimen4\font\relax}
\providecommand{\BIBforeignlanguage}[2]{{%
\expandafter\ifx\csname l@#1\endcsname\relax
\typeout{** WARNING: IEEEtran.bst: No hyphenation pattern has been}%
\typeout{** loaded for the language `#1'. Using the pattern for}%
\typeout{** the default language instead.}%
\else
\language=\csname l@#1\endcsname
\fi
#2}}
\providecommand{\BIBdecl}{\relax}
\BIBdecl

\bibitem{OlfatiSaber2007}
R.~Olfati-Saber, J.~A. Fax, and R.~M. Murray, ``Consensus and cooperation in
  networked multi-agent systems,'' \emph{Proceedings of the IEEE}, vol.~95,
  no.~1, pp. 215--233, Jan 2007.

\bibitem{Oh2015}
K.-K. Oh, M.-C. Park, and H.-S. Ahn, ``A survey of multi-agent formation
  control,'' \emph{Automatica}, vol.~53, pp. 424 -- 440, 2015.

\bibitem{Hatanaka2015}
T.~Hatanaka, N.~Chopra, M.~Fujita, and M.~Spong, \emph{Passivity-Based Control
  and Estimation in Networked Robotics}, 1st~ed., ser. Communications and
  Control Engineering.\hskip 1em plus 0.5em minus 0.4em\relax Springer
  International Publishing, 2015.

\bibitem{DePersis2018}
C.~De~Persis and N.~Monshizadeh, ``Bregman storage functions for microgrid
  control,'' \emph{IEEE Transactions on Automatic Control}, vol.~63, no.~1, pp.
  53--68, 2018.

\bibitem{Antsaklis2013}
P.~J. Antsaklis, B.~Goodwine, V.~Gupta, M.~J. McCourt, Y.~Wang, P.~Wu, M.~Xia,
  H.~Yu, and F.~Zhu, ``Control of cyberphysical systems using passivity and
  dissipativity based methods,'' \emph{European Journal of Control}, vol.~19,
  no.~5, pp. 379--388, 2013.

\bibitem{Arcak2007}
M.~Arcak, ``Passivity as a design tool for group coordination,'' \emph{IEEE
  Transactions on Automatic Control}, vol.~52, no.~8, pp. 1380--1390, 2007.

\bibitem{Chopra2006}
N.~Chopra and M.~W. Spong, \emph{Advances in Robot Control: From Everyday
  Physics to Human-Like Movements}.\hskip 1em plus 0.5em minus 0.4em\relax
  Springer, 2006, ch. Passivity-Based Control of Multi-Agent Systems, pp.
  107--134.

\bibitem{Stan2007}
G.-B. Stan and R.~Sepulchre, ``Analysis of interconnected oscillators by
  dissipativity theory,'' \emph{IEEE Transactions on Automatic Control},
  vol.~52, no.~2, pp. 256--270, Feb. 2007.

\bibitem{Tang2016}
Y.~Tang, Y.~Hong, and P.~Yi, ``Distributed optimization design based on
  passivity technique,'' in \emph{2016 12th IEEE International Conference on
  Control and Automation (ICCA)}, June 2016, pp. 732--737.

\bibitem{Khalil2002}
H.~Khalil, \emph{Nonlinear Systems}, ser. Pearson Education.\hskip 1em plus
  0.5em minus 0.4em\relax Prentice Hall, 2002.

\bibitem{Monzshizadeh2019}
N.~Monshizadeh, P.~Monshizadeh, R.~Ortega, and A.~van~der Schaft, ``Conditions
  on shifted passivity of port-hamiltonian systems,'' \emph{Systems \& Control
  Letters}, vol. 123, pp. 55 -- 61, 2019.

\bibitem{SimpsonPorco2019}
J.~W. {Simpson-Porco}, ``Equilibrium-independent dissipativity with quadratic
  supply rates,'' \emph{IEEE Transactions on Automatic Control}, vol.~64,
  no.~4, pp. 1440--1455, April 2019.

\bibitem{Pavlov2008}
A.~Pavlov and L.~Marconi, ``Incremental passivity and output regulation,''
  \emph{Systems \& Control Letters}, vol.~57, no.~5, pp. 400 -- 409, 2008.

\bibitem{Hines2011}
G.~H. Hines, M.~Arcak, and A.~K. Packard, ``Equilibrium-independent passivity:
  A new definition and numerical certification,'' \emph{Automatica}, vol.~47,
  no.~9, pp. 1949--1956, 2011.

\bibitem{Meissen2015}
C.~Meissen, L.~Lessard, M.~Arcak, and A.~K. Packard, ``Compositional
  performance certification of interconnected systems using admm,''
  \emph{Automatica}, vol.~61, pp. 55--63, 2015.

\bibitem{SimpsonPorco2016}
J.~W. Simpson-Porco, ``Input/output analysis of primal-dual gradient
  algorithms,'' in \emph{Proc. of the Annual Allerton Conference on
  Communication, Control, and Computing}, Allerton House, UIUC, Illinois, USA,
  2016, pp. 219--224.

\bibitem{Burger2014}
M.~B{\"u}rger, D.~Zelazo, and F.~Allg{\"o}wer, ``Duality and network theory in
  passivity-based cooperative control,'' \emph{Automatica}, vol.~50, no.~8, pp.
  2051--2061, 2014.

\bibitem{Rockafellar1998}
R.~T. Rockafellar, \emph{Network Flows and Monotropic Optimization}.\hskip 1em
  plus 0.5em minus 0.4em\relax Belmont, MA, USA: Athena Sci., 1998.

\bibitem{Sharf2017}
M.~Sharf and D.~Zelazo, ``A network optimization approach to cooperative
  control synthesis,'' \emph{IEEE Control Systems Letters}, vol.~1, no.~1, pp.
  86--91, 2017.

\bibitem{Sharf2018a}
M.~{Sharf} and D.~{Zelazo}, ``Analysis and synthesis of mimo multi-agent
  systems using network optimization,'' \emph{IEEE Transactions on Automatic
  Control}, vol.~64, no.~11, pp. 4512--4524, 2019.

\bibitem{Sharf2019e}
M.~Sharf, A.~Koch, D.~Zelazo, and F.~Allg{\"o}wer, ``Model-free practical
  cooperative control for diffusively coupled systems,'' \emph{arXiv preprint
  arXiv:1906.05204}, 2019.

\bibitem{Sharf2019f}
M.~Sharf and D.~Zelazo, ``{A Data-Driven and Model-Based Approach to Fault
  Detection and Isolation in Networked Systems},'' \emph{arXiv e-prints}, p.
  arXiv:1908.03588, Aug 2019.

\bibitem{Qu2014}
Z.~Qu and M.~A. Simaan, ``Modularized design for cooperative control and
  plug-and-play operation of networked heterogeneous systems,''
  \emph{Automatica}, vol.~50, no.~9, pp. 2405--2414, 2014.

\bibitem{Harvey2016}
R.~Harvey and Z.~Qu, ``Cooperative control and networked operation of
  passivity-short systems,'' in \emph{Control of Complex Systems: Theory and
  Applications}, K.~Vamvoudakis and S.~S.~Jagannathan, Eds.\hskip 1em plus
  0.5em minus 0.4em\relax Elsevier, 2016, pp. 499--518.

\bibitem{Trip2018}
S.~Trip and C.~De~Persis, ``Distributed optimal load frequency control with
  non-passive dynamics,'' \emph{IEEE Transactions on Control of Network
  Systems}, vol.~5, no.~3, pp. 1232--1244, 2018.

\bibitem{Xia2014}
M.~Xia, P.~J. Antsaklis, and V.~Gupta, ``Passivity indices and passivation of
  systems with application to systems with input/output delay,'' in \emph{IEEE
  Conference on Decision and Control}, Los Angeles, California, USA, 2014, pp.
  783--788.

\bibitem{Zhu2014}
F.~Zhu, M.~Xia, and P.~J. Antsaklis, ``Passivity analysis and passivation of
  feedback systems using passivity indices,'' in \emph{Proc. of the American
  Control Conference}, Portland, Oregon, USA, 2014, pp. 1833--1838.

\bibitem{Fradkov2003}
A.~Fradkov, ``Passification of non-square linear systems and feedback
  {Y}akubovich-{K}alman-{P}opov lemma,'' \emph{European Journal of Control},
  vol.~9, no.~6, pp. 577--586, 2003.

\bibitem{Byrnes1989}
C.~I. Byrnes and A.~Isidori, ``New results and examples in nonlinear feedback
  stabilization,'' \emph{Systems \& Control Letters}, vol.~12, no.~5, pp. 437
  -- 442, 1989.

\bibitem{Byrnes1991}
C.~I. Byrnes, A.~Isidori, and J.~C. Willems, ``Passivity, feedback equivalence,
  and the global stabilization of minimum phase nonlinear systems,'' \emph{IEEE
  Transactions on Automatic Control}, vol.~36, no.~11, pp. 1228--1240, 1991.

\bibitem{Fradkov1995}
A.~L. Fradkov, D.~Hill, Z.-P. Jiang, and M.~Seron, ``Feedback passification of
  interconnected systems,'' in \emph{IFAC NOLCOS}, vol.~2, 1995, pp. 660--665.

\bibitem{Jiang1996}
Z.-P. Jiang, D.~J. Hill, and A.~L. Fradkov, ``A passification approach to
  adaptive nonlinear stabilization,'' \emph{Systems \& Control Letters},
  vol.~28, no.~2, pp. 73 -- 84, 1996.

\bibitem{Fradkov1998}
A.~L. Fradkov and D.~J. Hill, ``Exponential feedback passivity and
  stabilizability of nonlinear systems,'' \emph{Automatica}, vol.~34, no.~6,
  pp. 697--703, 1998.

\bibitem{Sepulchre2012}
R.~Sepulchre, M.~Jankovic, and P.~V. Kokotovic, \emph{Constructive nonlinear
  control}.\hskip 1em plus 0.5em minus 0.4em\relax Springer Science \& Business
  Media, 2012.

\bibitem{Romer2019}
A.~Romer, J.~Berberich, J.~K{\"o}hler, and F.~Allg{\"o}wer, ``One-shot
  verification of dissipativity properties from input-output data,'' \emph{IEEE
  Control Systems Letters}, vol.~3, pp. 709--714, 2019.

\bibitem{Montenbruck2016}
J.~M. Montenbruck and F.~Allg{\"o}wer, ``Some problems arising in controller
  design from big data via input-output methods,'' in \emph{2016 IEEE 55th
  Annual Conference on Decision and Control (CDC)}, 2016, pp. 6525--6530.

\bibitem{Romer2017a}
A.~Romer, J.~M. Montenbruck, and F.~Allg{\"o}wer, ``Determining dissipation
  inequalities from input-output samples,'' in \emph{Proc.\ 20th IFAC World
  Congress}, 2017, pp. 7789--7794.

\bibitem{Xia2018}
M.~Xia, A.~Rahnama, S.~Wang, and P.~J. Antsaklis, ``Control design using
  passivation for stability and performance,'' \emph{IEEE Transactions on
  Automatic Control}, vol.~63, no.~9, pp. 2987--2993, 2018.

\bibitem{Jain2018}
A.~Jain, M.~Sharf, and D.~Zelazo, ``Regularization and feedback passivation in
  cooperative control of passivity-short systems: A network optimization
  perspective,'' \emph{IEEE Control Systems Letters}, vol.~2, no.~4, pp.
  731--736, 2018.

\bibitem{Godsil2001}
C.~Godsil and G.~Royle, \emph{Algebraic Graph Theory}, ser. Graduate Texts in
  Mathematics.\hskip 1em plus 0.5em minus 0.4em\relax Springer New York, 2001.

\bibitem{Rockafellar1997}
R.~T. Rockafellar, \emph{Convex Analysis}.\hskip 1em plus 0.5em minus
  0.4em\relax Princeton University Press, 1997.

\bibitem{Joo2016}
Y.~Joo, R.~Harvey, and Z.~Qu, ``Cooperative control of heterogeneous
  multi-agent systems in a sampled-data setting,'' in \emph{2016 IEEE 55th
  Conference on Decision and Control (CDC)}.\hskip 1em plus 0.5em minus
  0.4em\relax IEEE, 2016, pp. 2683--2688.

\bibitem{Atman2018}
M.~W.~S. Atman, T.~Hatanaka, Z.~Qu, N.~Chopra, J.~Yamauchi, and M.~Fujita,
  ``Motion synchronization for semi-autonomous robotic swarm with a
  passivity-short human operator,'' \emph{International Journal of Intelligent
  Robotics and Applications}, vol.~2, no.~2, pp. 235--251, 2018.

\bibitem{Sharf2019c}
M.~{Sharf} and D.~{Zelazo}, ``Network feedback passivation of passivity-short
  multi-agent systems,'' \emph{IEEE Control Systems Letters}, vol.~3, no.~3,
  pp. 607--612, July 2019.

\bibitem{Bondarko2003}
V.~A. Bondarko and A.~L. Fradkov, ``Necessary and sufficient conditions for the
  passivicability of linear distributed systems,'' \emph{Automation and Remote
  Control}, vol.~64, no.~4, pp. 517--530, 2003.

\bibitem{Selivanov2016}
A.~Selivanov, A.~Fradkov, and D.~Liberzon, ``Adaptive control of passifiable
  linear systems with quantized measurements and bounded disturbances,''
  \emph{Systems \& Control Letters}, vol.~88, pp. 62--67, 2016.

\bibitem{Montenbruck2015}
J.~M. Montenbruck, M.~B{\"u}rger, and F.~Allg{\"o}wer, ``Practical
  synchronization with diffusive couplings,'' \emph{Automatica}, vol.~53, pp.
  235 -- 243, 2015.

\bibitem{Franci2011}
A.~Franci, L.~Scardovi, and A.~Chaillet, ``An input-output approach to the
  robust synchronization of dynamical systems with an application to the
  hindmarsh-rose neuronal model,'' in \emph{2011 50th IEEE Conference on
  Decision and Control and European Control Conference}, Dec 2011, pp.
  6504--6509.

\bibitem{Dorfler2014}
F.~D{\"o}rfler and F.~Bullo, ``Synchronization in complex networks of phase
  oscillators: A survey,'' \emph{Automatica}, vol.~50, no.~6, pp. 1539 -- 1564,
  2014.

\bibitem{Bando1995}
M.~Bando, K.~Hasebe, A.~Nakayama, A.~Shibata, and Y.~Sugiyama, ``Dynamical
  model of traffic congestion and numerical simulation,'' \emph{Phys. Rev. E},
  vol.~51, pp. 1035--1042, Feb 1995.

\bibitem{Zames1966}
G.~{Zames}, ``On the input-output stability of time-varying nonlinear feedback
  systems part one: Conditions derived using concepts of loop gain, conicity,
  and positivity,'' \emph{IEEE Transactions on Automatic Control}, vol.~11,
  no.~2, pp. 228--238, April 1966.

\bibitem{McCourt2009}
M.~J. McCourt and P.~J. Antsaklis, ``Connection between the passivity index and
  conic systems,'' \emph{ISIS}, vol.~9, p. 009, 2009.

\bibitem{Xia2017}
M.~Xia, P.~J. Antsaklis, V.~Gupta, and F.~Zhu, ``Passivity and dissipativity
  analysis of a system and its approximation,'' \emph{IEEE Transactions on
  Automatic Control}, vol.~62, no.~2, pp. 620--635, 2017.

\bibitem{Fortuna1996}
L.~{Fortuna} and G.~{Muscato}, ``A roll stabilization system for a monohull
  ship: modeling, identification, and adaptive control,'' \emph{IEEE
  Transactions on Control Systems Technology}, vol.~4, no.~1, pp. 18--28, 1996.

\bibitem{Dorf2011}
R.~C. Dorf and R.~H. Bishop, \emph{Modern control systems}, 11th~ed.\hskip 1em
  plus 0.5em minus 0.4em\relax Pearson, 2011.

\bibitem{Scardovi2009}
L.~Scardovi and N.~E. Leonard, ``Robustness of aggregation in networked
  dynamical systems,'' in \emph{Proc. of the International Conference on Robot
  Communication and Coordination}, Odense, Denmark, 2009, pp. 1--6.

\end{thebibliography}

\appendix
\section{Proof of Theorem \ref{thm.GeometricApproach}}
The proof of Theorem~\ref{thm.GeometricApproach} is given below: 
\begin{proof}
Consider a PQI $a\xi^2 + b \xi\chi + c\chi^2 \ge 0$. If $a=c=0$ and $b\neq 0$, 
the solution set is either the union of the first and third quadrants, or the 
union of the second and fourth quadrants (depending whether $b>0$ or $b<0$). 
In particular, it is a symmetric double-cone in both these cases. Thus, we can 
assume that either $a\neq 0$ or $c \neq 0$. By switching the roles of $\xi$ and 
$\chi$, we may assume, without loss of generality, that $c\neq 0$. Note that if 
$(\xi,\chi)$ is a solution of the PQI, and $\lambda\in\R$, then 
$(\lambda\xi,\lambda\chi)$ is also a solution of the PQI. Thus, it's enough to 
show that the intersection of the solution set with the unit circle is a 
symmetric section. Writing a general point in $\mathbb{S}^1$ as $(\cos 
\theta,\sin \theta)$, the inequality becomes: 
\begin{align} \label{eq.IneqTrigCos}
a\cos^2\theta + b \cos\theta\sin\theta + c\sin^2\theta \ge 0.
\end{align}
We assume, for a moment, that $\cos \theta \neq 0$, and divide by $\cos^2\theta$, so that the inequality becomes:
\begin{align}\label{eq.IneqTrig}
a + b \tan\theta + c\tan^2 \theta \ge 0.
\end{align}
We denote $t_{\pm} = \frac{-b\pm\sqrt{b^2-4ac}}{2c}$ and consider two possible scenarios:
\begin{itemize}
\item $c<0$: In that case, \eqref{eq.IneqTrig} holds only when $\tan \theta$ is between $t_+$ and $t_-$. As $\tan$ is a monotone ascending function in $(-\pi/2,\pi/2)$ and $(\pi/2,1.5\pi)$, and tends to infinite values at the limits of said intervals, we conclude that \eqref{eq.IneqTrig} holds only when $\theta$ is inside $I_1\cup I_2$, where $I_1,I_2$ are the closed intervals which are the image of $[t_{-},t_{+}]$ under $\arctan(x)$ and $\arctan(x)+\pi$, so that $I_2 = I_1 + \pi$. Note that as $c<0$, any point at which $\cos \theta = 0$ does not satisfy \eqref{eq.IneqTrigCos}. Thus the intersection of the solution set of the PQI $a\xi^2 + b \xi\chi + c\chi^2 \ge 0$ with $\mathbb{S}^1$ is a symmetric section.
\item $c>0$: In that case, \eqref{eq.IneqTrig} holds only when $\tan \theta$ is outside the interval $(t_-,t_+)$. Similarly to the previous case, $\tan \theta \in (t_-,t_+)$ can be written as $B\cup (B+\pi)$ where $B$ is an \emph{open} section of angle $<\pi$. Thus its complement, which is the intersection of the solution set of the PQI $a\xi^2 + b \xi\chi + c\chi^2 \ge 0$ with $\mathbb{S}^1$, is a symmetric section.
\end{itemize}

Conversely, consider a symmetric double-cone $A$, and let $S=B\cup (B+\pi)$ be 
the associated symmetric section. Let $C\cup (C+\pi)$ be the complement of $S$ 
inside $\mathbb{S}^1$, where $C$ is an open section. We first claim that $\cos 
\theta \neq 0$ either on $B$ or on $C$. Indeed, $B\cup C$ is a half-open 
half-circle, and the only points at which $\cos \theta = 0$ are $\theta = \pm 
\pi/2$. Thus, $B \cup C$ can only contain one of them. Moreover, $B$ and $C$ 
are disjoint, so at least one does not include points at which $\cos \theta 
\neq 0$. Now, we consider two possible cases.

\begin{itemize}
\item $B$ (hence $S$) contains no points at which $\cos \theta = 0$. Then $\tan$ maps $B$ continuously into some interval $I=[t_-,t_+]$. Thus $\theta \in S$ if and only if $-(\tan \theta - t_-)(\tan \theta - t_+) \ge 0$. Inverting the process from the first part of the proof, the last inequality (which defines $S$) can be written as the intersection of the solution set of some PQI with $\mathbb{S}^1$. Thus $A$ is the solution set of the said PQI. Non triviality follows from the fact that $t_\pm$ are two distinct solutions to the associated equation.
\item $C$ contains no points at which $\cos \theta = 0$. Then $\tan$ maps $C$ continuously into some interval $I=(t_-,t_+)$. Thus, $\theta \in C\cup (C+\pi)$ if and only if $(\tan \theta - t_-)(\tan \theta - t_+) < 0$. Equivalently, $\theta \in S$ if and only if $(\tan \theta - t_-)(\tan \theta -t_+) \ge 0$. We can now repeat the argument for the first case to conclude that $A$ is the solution set of a non-trivial PQI.
\end{itemize}

As for uniqueness, suppose the non-trivial PQIs $a_1 \xi^2 + b_1 \xi \chi + c_1 \chi^2 \ge 0$ and $a_2 \xi^2 + b_2 \xi \chi + c_2 \chi^2 \ge 0$ define the same solution set. Then the equations $a_1 \xi^2 + b_1 \xi \chi + c_1 \chi^2 = 0$ and $a_2 \xi^2 + b_2 \xi \chi + c_2 \chi^2 = 0$ have the same solutions (as the boundaries of the solution sets). Assume first that either $a_1\neq 0$ or that $a_2\neq 0$.  In particular, for $\xi = \tau\chi$, both equations $\chi^2(a_1\tau^2 + b_1 \tau + c_1 ) = 0$ and $\chi^2(a_2\tau^2 +b_2\tau + c_2)=0$ have the same solutions. Dividing by $\chi^2$ implies both equations have two solutions, $t_-\neq t_+$, as $b_1^2-4a_1c_1 >0$ and $b_2^2-4a_2c_2 > 0$. Thus, we can write:
\begin{align*}
a_1\tau^2+b_1\tau+c_1 = a_1(\tau - t_-)(\tau - t_+),\\\ a_2\tau^2 +b_2\tau + c_2 = a_2(\tau - t_-)(\tau - t_+).
\end{align*}
implying the original PQIs are the same up to scalar, which must be positive 
due to the direction of the inequalities. 

Otherwise, $a_1 = a_2 = 0$, so we must have $b_1,b_2\neq 0$, as otherwise $b_1^2-4a_1c_1 = 0$ or $b_2^2 - 4a_2c_2 = 0$. Plugging $\chi = 1$, we get that the equations $b_1\xi + c_1 = 0$ and $b_2\xi + c_2 = 0$ have the same solutions, implying that $(b_1,c_1)$ and $(b_2,c_2)$ are equal up to a multiplicative scalar. As $a_1=a_2=0$, we conclude the same about 
the original PQIs. Moreover, the scalar has to be positive due to the direction of the original PQIs. This completes the proof. 
\end{proof}

\end{document}